\documentclass[sigconf,nonacm]{acmart}

\usepackage{amsmath}
\usepackage{xcolor}
\usepackage{balance}
\usepackage{float} 

\usepackage{amssymb}
\usepackage{bm}
\usepackage{nicefrac}
\usepackage{booktabs}
\usepackage{array}
\usepackage{multirow}
\usepackage{threeparttable}
\usepackage{makecell}
\usepackage[procnumbered,ruled,vlined,linesnumbered]{algorithm2e}
\usepackage{stfloats}
\usepackage{graphicx}
\usepackage{subfigure}
\usepackage{hyperref}
\usepackage{enumerate}
\def\norm#1{\| #1 \|}

\newtheorem{problem}{Problem}
\newtheorem{theorem}{Theorem}[section]

\newtheorem{lemma}[theorem]{Lemma}

\newtheorem{definition}[theorem]{Definition}

\newcommand{\bb}{\bm{b}}    

\newenvironment{fminipage}%
{\begin{Sbox}\begin{minipage}}%
		{\end{minipage}\end{Sbox}\fbox{\TheSbox}}

\def\abs#1{\left|#1  \right|}

\def\trace#1{\mathrm{Tr} \left(#1 \right)}

\newfont{\nset}{msbm10}

\def\kh#1{\left( #1 \right)}

\newcommand{\removelatexerror}{\let\@latex@error\@gobble}

\newcommand\LL{\bm{\mathit{L}}}

\def\trace#1{\mathrm{Tr} \left(#1 \right)}

\def\aa{\pmb{\mathit{a}}}
\newcommand{\one}{\mathbf{1}}

\newcommand\XX{\boldsymbol{\mathit{X}}}
\newcommand\yy{\boldsymbol{\mathit{y}}}
\newcommand\zz{\boldsymbol{\mathit{z}}}
\newcommand\xx{\boldsymbol{\mathit{x}}}
\newcommand\aaa{\boldsymbol{\mathit{a}}}

\newcommand\ee{\boldsymbol{\mathit{e}}}

\newcommand\vvv{\boldsymbol{\mathit{v}}}
\newcommand\hh{\boldsymbol{\mathit{h}}}

\renewcommand\AA{\boldsymbol{\mathit{A}}}
\newcommand\BB{\boldsymbol{\mathit{B}}}

\newcommand\JJ{\boldsymbol{\mathit{J}}}
\newcommand\DD{\boldsymbol{\mathit{D}}}

\newcommand\MM{\boldsymbol{\mathit{M}}}

\newcommand\ZZ{\boldsymbol{\mathit{Z}}}

\newcommand\QQ{\boldsymbol{\mathit{Q}}}

\newcommand\II{\boldsymbol{\mathit{I}}}
\newcommand\uu{\boldsymbol{\mathit{u}}}

\newcommand{\SDDMSolver}{\textsc{Solve}}
\newcommand\ZZtil{\widetilde{\boldsymbol{\mathit{Z}}}}
\newcommand\zztil{\widetilde{\boldsymbol{\mathit{z}}}}

\DeclareMathOperator*{\argmin}{arg\,min}

\DontPrintSemicolon
\SetKw{KwAnd}{and}
\SetFuncSty{textsc}
\SetKwInOut{Input}{Input\ \ \ \ }
\SetKwInOut{Output}{Output}

\usepackage{tabularx}
\usepackage{stfloats}

\begin{document}

\title{Minimizing Total Biharmonic Distance in Large Graphs via Link Recommendation}

\author{Xinna Zhou}
\affiliation{%
  \department{College of Computer Science and Artificial Intelligence}
  \institution{Fudan University}
  \city{Shanghai}
  \country{China}
  \postcode{200433}
}
\email{23210240097@m.fudan.edu.cn}

\author{Zhongzhi Zhang}
\authornote{Corresponding author.}
\affiliation{
  \department{College of Computer Science and Artificial Intelligence}
 \institution{Fudan University}
 \city{Shanghai}
 \country{China}
}
\email{zhangzz@fudan.edu.cn}

\begin{abstract}
The total biharmonic distance, which is the sum of the biharmonic distance between every pair of nodes in a network, is a key metric for evaluating network connectivity and robustness. In this paper, we study the problem of minimizing the total biharmonic distance by adding $k$ nonexistent edges for a given graph $G$ and budget $k$. The problem is computationally challenging. We show that the objective function of the problem is monotone but not supermodular. To solve this problem, we propose simple greedy algorithms with cubic time complexity. To mitigate the high time complexity of these greedy algorithms, we apply several techniques, including the projection method, the Laplacian solver, and convex hull approximation. These techniques reduce the time complexity of our proposed algorithms from cubic to nearly linear while providing error guarantees. Finally, extensive experiments on real datasets demonstrate both the efficiency and effectiveness of our proposed algorithms.

\vspace{4pt}
This paper has been published in Proceedings of the 32nd ACM SIGKDD Conference on Knowledge Discovery and Data Mining V.1 (KDD '26).
DOI: \url{https://doi.org/10.1145/3770854.3780204}.
\end{abstract}



\maketitle

\section{Introduction}
Graph distance metrics play a crucial role in the realm of network analysis. Notable distance metrics include geodesic distance~\cite{newman2018networks} and resistance distance~\cite{KlRa93}. Resistance distance originates from electrical network theory and has been widely utilized in clustering~\cite{SheMalik2000SpecClust}, graph learning~\cite{kreuzer2021rethinking,black2023understanding,ZhuEtal2003SemiSupHarmonic}, and collaborative recommendation systems~\cite{FoPiReSa07, KuSc07}. Additionally, it has been applied to ranking problems in large graphs~\cite{SaMoPr08}, detecting anomalous changes in graph structure over time~\cite{SrDa14}, assessing node significance~\cite{lu2024resistance}, and enhancing the efficiency of graph systems~\cite{QiDhTaPeWa21}. Over the past few decades, considerable efforts have been dedicated to further understanding and exploring the properties of resistance distance~\cite{DoBu13, DoSibu18,ThYaNa19}.

Resistance distance has several variants~\cite{verma2017hunt}, among which biharmonic distance~\cite{lipman2010biharmonic} stands out. It was introduced to capture more nuanced local and global structural information, offering advantages over resistance distance in many applications, such as maintaining graph connectivity and identifying critical cut edges~\cite{blackbiharmonic}. Biharmonic distance has been applied across various domains, including graph clustering~\cite{blackbiharmonic}, graph embedding techniques~\cite{kreuzer2021rethinking,black2023understanding}, graph matching~\cite{fan2020spectral}, network coherence~\cite{Yi2022BiharmonicDP,zhang2020fast}, leader selection in noisy networks~\cite{bd14}, identifying crucial edges~\cite{YiShLiZh18}, and other network science problems~\cite{bamieh2012coherence,tyloo2019key}.

Beyond individual metrics, the aggregate of a metric in a graph is often used to evaluate the overall network. For instance, the Kirchhoff index sums the resistance distances between all node pairs~\cite{KlRa93,GhBoSa08}. Similarly, the total biharmonic distance, or biharmonic index~\cite{Yi2022BiharmonicDP}, is the sum of biharmonic distances across all node pairs in the graph. These aggregate metrics serve as powerful tools for evaluating various properties of a network and have attracted significant attention. The Kirchhoff index is key for measuring network connectedness~\cite{TiLe10} and evaluating the robustness of consensus algorithms~\cite{QiZhYiLi19}. The total biharmonic distance can be expressed by the Kirchhoff index and the resistance distances between some node pairs~\cite{Tyloo2017RobustnessOS}. {
The total biharmonic distance has many application scenarios. It can effectively evaluate the robustness of second-order consensus algorithms in distributed network systems~\cite{Tyloo2017RobustnessOS,XuWuZhZhKaCh22} and identify local vulnerabilities~\cite{tyloo2019key}. For example, in power grids and complex oscillator networks, total biharmonic distance can measure the fault tolerance of the system~\cite{tyloo2019key}. In communication networks, when lines are cut and information must be redirected, total biharmonic distance can evaluate a network’s vulnerability to changes~\cite{Tyloo2017RobustnessOS}. In addition, it is an important tool for evaluating network connectivity~\cite{blackbiharmonic} and can also be used to define edge importance metrics~\cite{siami2017growing}.}

Networks with smaller total biharmonic distance are known to exhibit enhanced robustness and stronger connectivity, leading to improved overall system performance~\cite{bamieh2012coherence,Yi2022BiharmonicDP}. Therefore, optimizing total biharmonic distance through different graph operations becomes essential. Motivated by practical scenarios such as link recommendation in online social networks, we study the problem of optimizing (minimizing) total biharmonic distance by adding nonexistent edges (links) to a graph. More precisely, we aim to devise effective and efficient algorithms for the following optimization problem: For a given connected undirected unweighted graph $G=(V,E)$ with $n$ nodes and $m$ edges, a small positive integer $k$, how to add $k$ nonexistent edges from a candidate edge set $Q=(V \times V) \backslash E$ to graph $G$, so that the total biharmonic distance for the resulting graph is minimized.

The main contributions of this paper are summarized below.

$\bullet$ We show that the total biharmonic distance is monotonically decreasing but not supermodular. We then propose a simple greedy algorithm with cubic time complexity, providing an approximation guarantee based on the curvature and submodularity ratio of the objective function. We also present a gradient-based greedy algorithm as a novel approach to address this problem.

$\bullet$ We employ several techniques, including random projection, the Laplacian solver, and convex hull approximation, resulting in a scalable and efficient algorithm \textsc{ApproxFast}. This algorithm is gradient-based and achieves a time complexity of \(\tilde{O}(k(nl+m)/\epsilon^2)\), where $\tilde{O}(\cdot)$ hides poly$(\log (n))$ factors, $\epsilon > 0$ is the error parameter that balances performance and efficiency, and $l$ represents the number of nodes on the boundary of the approximate convex hull, which is typically small in most real-world networks.

$\bullet$ We evaluate our algorithms by performing experiments on many real-world networks of varying scales, showing that our proposed fast algorithm \textsc{ApproxFast} achieves both good efficiency and effectiveness and is scalable to large networks with over four million nodes.

\section{Preliminaries}
\subsection{Notations}
We use $\mathbb{R}$ to denote real number field. We utilize bold lowercase (e.g., $\xx$) for vectors and uppercase (e.g., $\MM$) for matrices. Elements are indicated by subscripts, e.g., $\xx_i$ or $\MM_{i,j}$. For a matrix $\MM$, $\MM[i,:]$ and $\MM[:,j]$ represent the $i$-th row and $j$-th column, respectively. Let $\ee_i$ denote the $i$-th standard basis vector (1 at position $i$, 0 elsewhere). We use $\mathbf{1}$ (resp. $\JJ$) to denote the all-ones (column) vector (resp. matrix). We write $\xx^\top$ and $\MM^\top$to represent the transposes of vector $\xx$ and matrix $\MM$, respectively. For any vector $\aaa$, we use $\norm{\aaa}=\sqrt{\sum\nolimits_i \aaa_i^2}$ to denote the $\ell_2$ norm of $\aaa$, and use $\norm{\aaa}_{\XX} = \sqrt{\aaa^\top \XX \aaa}$ to denote the matrix norm of $\aaa$ for matrix $\XX$. For two non-negative scalars $a$ and $b$, we use $a \approx_{\epsilon} b$ to denote that $a$ is an $\epsilon$-approximation of $b$ obeying relation $(1-\epsilon) b \leq a \leq(1+\epsilon) b$. For two matrices $\boldsymbol{A}\mathrm{~and~}\boldsymbol{B}$, we write $\boldsymbol{A}\preceq\boldsymbol{B}$ to denote that $\boldsymbol{B}-\boldsymbol{A}$ is positive semi-definite, i.e., $\boldsymbol{x}^T\boldsymbol{A}\boldsymbol{x}\leq\boldsymbol{x}^T\boldsymbol{B}\boldsymbol{x}$ holds for every vector $\boldsymbol{x}.$

\subsection{Graph and Related Matrices}
Let \( G = (V, E) \) be a connected undirected unweighted graph with \( V = \{1,2,\dots,n\} \) and \( E = \{e_1, e_2, \dots, e_m\} \). The adjacency matrix \( \AA \in \{0,1\}^{n \times n} \) satisfies \( \AA_{i,j} = 1 \) if nodes \( i \) and \( j \) are adjacent, and \( 0 \) otherwise. The degree matrix is \( \DD = \text{diag}(d_1, \dots, d_n) \), where \( d_i = \sum_{j=1}^n \AA_{i,j} \), and the Laplacian matrix is \( \LL = \DD - \AA \).  If we fix an arbitrary orientation for all edges in $G$, then we define the signed edge-vertex incidence matrix $\BB \in \mathbb{R}^{m \times n}$, whose entries are: $\BB_{e,u}=1$ if node $u$ is the head of edge $e$, $\BB_{e,u}=-1$ if node $u$ is the tail of edge $e$, and $\BB_{e,u}=0$ otherwise. $\LL$ can also be represented by $\LL=\BB^T\BB$. For any pair of distinct nodes $u,v\in V$, we define $\bb_{uv} = \ee_{u}-\ee_{v}$. If an edge \( e \) connects \( u \) and \( v \) and we fix an arbitrary orientation such that \( e \) is directed from \( u \) to \( v \), we define $\bb_{e}=\ee_{u}-\ee_{v}$.  Since \( \LL \) is symmetric, it has eigenvalues \( 0 = \lambda_1(\LL) < \lambda_2(\LL) \leq \cdots \leq \lambda_n(\LL) \) with corresponding orthonormal eigenvectors \( \uu_1,\uu_2, \dots, \uu_n \), and can be written as \( \LL = \sum_{i=2}^{n} \lambda_i \uu_i \uu_i^T \). Its Moore-Penrose pseudoinverse is \( \LL^\dag = (\LL + \frac{1}{n} \JJ)^{-1} - \frac{1}{n} \JJ \). We have \( \frac{1}{2n^4} \LL \LL^\dag \preceq \LL \preceq n \II \) and \( \frac{1}{n} \LL \LL^\dag \preceq \LL^\dag \preceq 2n^4 \II \)~\cite{Ch97}. The eigenvalues and the corresponding eigenvectors of $\LL^2$ are $\lambda_i^2$ and $\boldsymbol{u}_{i},i=1,2,\ldots,n.$ Thus, $\LL^2$ has the following spectral decomposition: $\boldsymbol{\LL}^{2}= \sum\nolimits_{i= 2}^{n}\lambda _{i}^{2}(\LL)\boldsymbol{u}_{i}\boldsymbol{u}_{i}^{\top }$, which indicates that $\LL^2$ is positive semi-definite with only one zero eigenvalue. Its Moore-Penrose pseudo-inverse $\LL^{2\dag}$ is positive semi-definite with only one zero eigenvalue. Similarly, $\LL^3$ and $\LL^{3\dag}$ are positive semi-definite with only one zero eigenvalue. For any edge set \( T \subseteq (V \times V) \setminus E \), let \( G_T = (V, E \cup T) \) be the graph after adding edges in \( T \). 

\subsection{ Biharmonic Distances and Related Metrics}

\begin{definition}[\cite{YiShLiZh18}]\label{def:biharmonic}  For a graph $G=(V, E)$, the biharmonic distance $\bb(u, v)$ between any pair of distinct nodes $u,v$ is defined by
    \begin{align*}
        \bb(u,v)=\sqrt{\bb_{uv}^{\top}\LL^{2\dag}\bb_{uv}}=\|\LL^{\dag}\bb_{uv}\| = \sqrt{\LL_{u,u}^{2\dag} + \LL_{v,v}^{2\dag} - 2\LL_{u,v}^{2\dag}}.
    \end{align*}
\end{definition}
We denote the square of the biharmonic distance as \(\bb^2(u,v)\). For brevity, we will sometimes refer to \(\bb^2(u,v)\) as the biharmonic distance, when the context makes it clear that we mean its square. Based on the biharmonic distance, we can derive the following graph index.

\begin{definition}[\cite{YiYaZhPa18}]\label{def:total biharmonic distance}
    For a graph $G=(V, E)$, the total biharmonic distance $B(G)$ is defined by 
    \begin{equation*}
        B(G)= \frac12\sum\nolimits_{u\in V}\sum\nolimits_{\nu\in V}\bb^2(u,v)=n\sum\nolimits_{i=2}^n\frac1{\lambda_i^2}=n\trace{\boldsymbol{\LL^{2\dagger}}}.
    \end{equation*}
\end{definition}

The biharmonic distance can be further generalized to higher-order variants~\cite{siami2017growing} by considering arbitrary powers of $\LL^{\dagger}$.  
\begin{definition}[\cite{blackbiharmonic}]
    For a graph $G=(V, E)$, the $k$-harmonic distance $\hh_k(u, v)$ between any pair of distinct nodes $u,v$ is defined by $$\hh_k(u,v)=\sqrt{\bb_{uv}^{\top}\LL^{k\dag}\bb_{uv}}.$$
\end{definition}
Similarly, we will refer to \(\hh_k^2(u,v)\) as the \(k\)-harmonic distance, rather than the square of the \(k\)-harmonic distance, when the context is clear. Notably, when $k=1$, $\hh_k^2(u,v)=R_{uv}$ which is the resistance distance between nodes $u$ and $v$. When $k=2$, $\hh_k(u,v)=\bb(u,v)$ which is the biharmonic distance between nodes $u$ and $v$. We also introduce the total $k$-harmonic distance.
\begin{definition}[\cite{tyloo2019key}]
    For a graph $G=(V, E)$, the total $k$-harmonic distance $K(G)$ is defined by 
    \begin{equation*}
        K(G)= \frac12\sum\nolimits_{u\in V}\sum\nolimits_{\nu\in V}\hh_k^2(u,v)=n\sum\nolimits_{i=2}^n\frac1{\lambda_i^k}=n\trace{\boldsymbol{\LL^{k\dagger}}}.
    \end{equation*}
\end{definition}

\section{Problem Formulation}
In this section, we begin by examining the monotonicity of the total biharmonic distance. Next, we present the problem statement. We then analyze the properties of the objective function before introducing two simple greedy algorithms.

\subsection{Monotonicity of Total Biharmonic Distance}

The monotonically decreasing property of the total biharmonic distance was shown in \cite{wei2021biharmonic}. However, the method in \cite{wei2021biharmonic} does not give the marginal decrease, which is important for our algorithm. We take a different approach by first deriving the marginal decrease in total biharmonic distance when an edge is added.

\begin{lemma}\label{marginal} For a connected graph $G=(V,E)$, the marginal decrease of the total biharmonic distance after adding an edge $e$ is$$\Delta(e)=B(G)-B(G_{\{e\}})=-n\frac{(\boldsymbol{b}_e^\top\boldsymbol{L}^{2\dagger}\boldsymbol{b}_e)^2}{(1+\boldsymbol{b}_e^\top\boldsymbol{L}^\dagger\boldsymbol{b}_e)^2}+n\frac{2 \boldsymbol{b}_e^\top\boldsymbol{L}^{3\dagger}\boldsymbol{b}_e}{1+\boldsymbol{b}_e^\top\boldsymbol{L}^\dagger\boldsymbol{b}_e}.$$
\end{lemma}
\begin{proof}
By Definition~\ref{def:total biharmonic distance},  $\Delta(e)=n\trace{\boldsymbol{L^{2\dagger}}}-n\trace{(\boldsymbol{L}+\boldsymbol{b}_e\boldsymbol{b}_e^\top)^{2\dagger}}.$ Using the Sherman-Morrison formula~\cite{Me73}, we have 
\begin{equation}\label{morrison}
    (\LL + \bb_e \bb_e^T)^{\dag} = \LL^{\dag} - \frac{\LL^{\dag} \bb_e \bb_e^T \LL^{\dag}}{1 + \bb_e^T \LL^{\dag} \bb_e}.
\end{equation}
Squaring both sides, we obtain
\begin{equation}\label{Lsquare}
\begin{aligned}
&(\LL + \bb_e \bb_e^T)^{2\dag}= \left(\LL^{\dag} - \frac{\LL^{\dag} \bb_e \bb_e^T \LL^{\dag}}{1 + \bb_e^T \LL^{\dag} \bb_e}\right)^2\\
&=\LL^{2\dag} -  \frac{\LL^{2\dag} \bb_e \bb_e^T \LL^{\dag}}{1 + \bb_e^T \LL^{\dag} \bb_e}- 
 \frac{\LL^{\dag} \bb_e \bb_e^T \LL^{2\dag}}{1 + \bb_e^T \LL^{\dag} \bb_e} + \frac{(\LL^{\dag} \bb_e \bb_e^T \LL^{\dag})^2}{(1 + \bb_e^T \LL^{\dag} \bb_e)^2}. 
\end{aligned}
\end{equation}
Thus, we have
\begin{equation*}
\begin{aligned}
\trace{(\boldsymbol{L}+\boldsymbol{b}_e\boldsymbol{b}_e^\top)^{2\dagger}}&=\trace{\boldsymbol{L^{2\dagger}}}-\frac{2 \trace{\boldsymbol{L^{2\dagger}}\boldsymbol{b}_e\boldsymbol{b}_e^\top \boldsymbol{L}^\dagger} }{1+\boldsymbol{b}_e^\top\boldsymbol{L}^\dagger\boldsymbol{b}_e}\\&+\frac{\trace{\boldsymbol{L}^\dagger\boldsymbol{b}_e\boldsymbol{b}_e^\top \boldsymbol{L^{2\dagger}} \boldsymbol{b}_e\boldsymbol{b}_e^\top \boldsymbol{L}^\dagger}}{(1+\boldsymbol{b}_e^\top\boldsymbol{L}^\dagger\boldsymbol{b}_e)^2}.
\end{aligned}
\end{equation*}
Therefore, one can conclude that
\begin{equation}\label{difference}
\begin{aligned}
\Delta(e) &= -n\frac{\trace{\boldsymbol{L}^\dagger\boldsymbol{b}_e\boldsymbol{b}_e^\top \boldsymbol{L^{2\dagger}} \boldsymbol{b}_e\boldsymbol{b}_e^\top \boldsymbol{L}^\dagger}}{(1+\boldsymbol{b}_e^\top\boldsymbol{L}^\dagger\boldsymbol{b}_e)^2}+
n\frac{2 \trace{\boldsymbol{L^{2\dagger}}\boldsymbol{b}_e\boldsymbol{b}_e^\top \boldsymbol{L}^\dagger} }{1+\boldsymbol{b}_e^\top\boldsymbol{L}^\dagger\boldsymbol{b}_e}\\
&=-n\frac{\boldsymbol{b}_e^\top\boldsymbol{L}^{2\dagger}\boldsymbol{b}_e \trace{\boldsymbol{L}^{\dagger}\boldsymbol{b}_e \boldsymbol{b}_e^\top \boldsymbol{L}^{\dagger}}}{(1+\boldsymbol{b}_e^\top\boldsymbol{L}^\dagger\boldsymbol{b}_e)^2}+n\frac{2 \trace{\boldsymbol{b}_e^\top\boldsymbol{L}^{3\dagger}\boldsymbol{b}_e }}{1+\boldsymbol{b}_e^\top\boldsymbol{L}^\dagger\boldsymbol{b}_e}\\
&=-n\frac{(\boldsymbol{b}_e^\top\boldsymbol{L}^{2\dagger}\boldsymbol{b}_e)^2}{(1+\boldsymbol{b}_e^\top\boldsymbol{L}^\dagger\boldsymbol{b}_e)^2}+n\frac{2 \boldsymbol{b}_e^\top\boldsymbol{L}^{3\dagger}\boldsymbol{b}_e}{1+\boldsymbol{b}_e^\top\boldsymbol{L}^\dagger\boldsymbol{b}_e}.\\
\end{aligned}
\end{equation}
This finishes the proof.
\end{proof}

Building on Lemma~\ref{marginal}, we now prove the monotonicity of total biharmonic distance in Lemma~\ref{decreasing}.
\begin{lemma}\label{decreasing}
    Let $G=(V,E)$ be a connected graph with $n$ nodes and $e\in (V\times V) \backslash E$ . Then
$B(G_{\{e\}})<B(G)$, that is, $B(G)$ is monotonically decreasing. 
\end{lemma}
\begin{proof}Suppose $\uu =\BB\LL^{2\dagger}\bb_e$, $\vvv =\BB\LL^{\dagger}\bb_e,$ thus 
$$\uu^T\uu=\bb_e^T\LL^{2\dagger}\BB^T \BB\LL^{2\dagger}\bb_e=\bb_e^T\LL^{2\dagger}\LL\LL^{2\dagger}\bb_e=\bb_e^T\LL^{3\dagger}\bb_e.$$ Likewise, $\vvv^T\vvv=\bb_e^T\LL^{\dagger}\bb_e$, $\uu^T\vvv=\bb_e^T\LL^{2\dagger}\bb_e.$  By Cauchy-Schwarz inequality~\cite{steele2004cauchy}, we have $(\uu^T\vvv)^2\leq (\uu^T\uu) (\vvv^T\vvv)$, implying
 $$
 (\bb_e^T\LL^{2\dagger}\bb_e)^2 \leq  \bb_e^T\LL^{3\dagger}\bb_e  \bb_e^T\LL^{\dagger}\bb_e.
 $$
 
According to Lemma~\ref{marginal}, we have
 $$
\begin{aligned}
 \Delta(e) &\geq -n\frac{\bb_e^T\LL^{3\dagger}\bb_e  \bb_e^T\LL^{\dagger}\bb_e}{(1+\boldsymbol{b}_e^\top\boldsymbol{L}^\dagger\boldsymbol{b}_e)^2}+n\frac{2 \boldsymbol{b}_e^\top\boldsymbol{L}^{3\dagger}\boldsymbol{b}_e(1+\boldsymbol{b}_e^\top\boldsymbol{L}^\dagger\boldsymbol{b}_e)}{(1+\boldsymbol{b}_e^\top\boldsymbol{L}^\dagger\boldsymbol{b}_e)^2}\\
 &=n\frac{2 \boldsymbol{b}_e^\top\boldsymbol{L}^{3\dagger}\boldsymbol{b}_e+\boldsymbol{b}_e^\top\boldsymbol{L}^{3\dagger}\boldsymbol{b}_e\boldsymbol{b}_e^\top\boldsymbol{L}^\dagger\boldsymbol{b}_e}{(1+\boldsymbol{b}_e^\top\boldsymbol{L}^\dagger\boldsymbol{b}_e)^2} >0.
\end{aligned}
 $$ 
This establishes that \( \Delta(e) > 0 \), thereby completing the proof.
\end{proof}

This lemma indicates that adding any nonexistent edge to \(G\) decreases the total biharmonic distance. Since a smaller biharmonic distance indicates better performance, it is useful in applications like assessing the robustness of second-order consensus algorithms in noisy networks~\cite{bamieh2012coherence,Yi2022BiharmonicDP}. This motivates our exploration of minimizing the total biharmonic distance by adding \(k\) new edges.

\subsection{Problem Statement}
We now give a mathematical formulation of our problem. 
\begin{problem}\label{problem}
Given a connected graph $G=(V,E)$, a candidate edge set $Q=(V\times V) \backslash E$ and an integer $k \leq |Q|$, find a subset $T \subseteq Q$ of $k$ edges such that the total biharmonic distance of the graph $G^*=(V,E\cup T^*)$ is minimized. More formally, the objective is to find
\begin{equation*}
T^* = \argmin_{T \subseteq Q,|T|=k} B(G_T).
\end{equation*}
\end{problem}

Problem~\ref{problem} is inherently a combinatorial problem, which is computationally challenging.

\subsection{Properties of Objective Function}
We have already proved that the objective function of Problem~\ref{problem} $B(\cdot)$ is monotone in Lemma~\ref{decreasing}. Next, we show that function $B(\cdot)$ is not supermodular.

\begin{lemma}~\label{supermodular}
The objective function in Problem~\ref{problem}, specifically the total biharmonic distance \( B(\cdot) \), is not supermodular. This means that there exist two edge sets \( S \) and \( T \) such that \( S \subset T \subseteq Q \), and an edge \( e \in Q \backslash T \), for which the following inequality holds
\[
B(G_{S \cup \{e\}}) - B(G_S) > B(G_{T \cup \{e\}}) - B(G_T).
\]

\end{lemma}
\begin{proof}
        To prove the non-supermodularity of the objective function, consider the graph shown in Figure~\ref{fig:supermodular}. The solid black lines denote the edges of the original graph, the solid blue line represents the edge in \(T\), and the dashed red line indicates the edge \(e\). We define \(S=\emptyset\), \(T=\{(1,3)\}\), and \(e=(1,4)\). The computation yields \(B(G_{S \cup \{e\}}) - B(G_S)=-0.89\) and \(B(G_{T \cup \{e\}}) - B(G_T)=-1.01\), violating the condition of supermodularity since \(-0.89 > -1.01\). Therefore, the objective function is non-supermodular.
\end{proof}
\begin{figure}[!hbtp]
		\centering
		\includegraphics[width=0.4\linewidth]{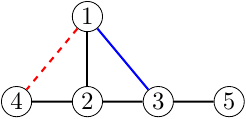}
		\caption{A counterexample with 5 nodes for Lemma~\ref{supermodular}. \label{fig:supermodular}}
\end{figure}

For traditional monotonically decreasing and supermodular functions, a simple greedy strategy can attain an approximation ratio of \((1-1/e)\) solution in each iteration~\cite{NeWoFi78}. Despite the objective function \(B(\cdot)\) being not supermodular, the greedy method can still remain a suitable strategy for solving Problem~\ref{problem} with guaranteed performance. We define the non-negative, monotone-increasing function \(f(T) = B(G) - B(G_T)\), where \(T \subseteq Q\). We introduce a new quantity \(\Theta_T(S) = f(S \cup T) - f(S)\), which represents the marginal benefit of adding the set \(T \subseteq Q\) to the set \(S \subseteq Q\). Utilizing this, we introduce the concepts of submodularity ratio and curvature.

\begin{definition}[\cite{BiBuKrTs17}]
For a non-negative set function \( f \), the submodularity ratio is defined as the largest scalar \(\gamma\) such that for any two arbitrary sets \(S \subseteq Q\) and \(T \subseteq Q\),
\begin{equation*}
\sum\nolimits_{i\in T\setminus S}\Theta_{i}(S)\geq\gamma\Theta_T(S). \Omega(n^2)
\end{equation*}
\end{definition}

\begin{definition}[\cite{BiBuKrTs17}]\label{curvature}
For a non-negative set function \(f\), the curvature is defined as the smallest scalar \(\alpha\) such that for any two arbitrary sets \(S \subseteq Q\) and \(T \subseteq Q\), and any element $j \in S\backslash T$,
\begin{equation*}
\Theta_j(S \cup T \backslash j)\geq(1-\alpha)\Theta_j(S\backslash j).
\end{equation*}
\end{definition}

We then establish the bounds for the submodularity ratio \(\gamma\) and curvature \(\alpha\) of the function \(f(T) = B(G) - B(G_T)\), which will be used in the next section.
 
\begin{lemma}\label{bound}
The submodularity ratio $\gamma$ of the function $f(T)=B(G)-B(G_T)$, $T \subseteq  Q$ is bounded by
$  1>\gamma\geq\left(\frac{\lambda_2(\boldsymbol{L})}n\right)^5,
$
and its curvature $\alpha$ is constrained by
$0<\alpha\leq1-\left(\frac{\lambda_2(\boldsymbol{L})}n\right)^5.
$
\end{lemma}
\begin{proof}Let $Q= ( V\times V) \backslash E$ be the candidate edge set, $S$ and $T$ be two arbitrary subsets of $E$, and $u,v$ be any two subsets of the candidate edge set $Q$. To begin with, 
we first derive a lower and upper bound for the marginal benefit function $\Theta_T(S)=f(S\cup {T})-f(S)$, respectively. We use $\boldsymbol{L}(S)$ to denote the Laplacian matrix of the resultant graph $G_S$. On the one hand,
$$
\begin{aligned}
\Theta_T(S)&=f(S\cup T)-f(S)=B(G_S)-B(G_{S\cup T}) \\
&= n\trace{\boldsymbol{L}(S)^{2\dagger}}-n\trace{\boldsymbol{L}(S\cup T)^{2\dagger}} \\
&= \sum\nolimits_{i=2}^{n}\frac n{\lambda_i^2(\boldsymbol{L}(S))}-\frac n{\lambda_i^2(\boldsymbol{L}(S\cup T))} \\
&=n\sum\nolimits_{i=2}^n\frac{(\lambda_i(\boldsymbol{L}(S\cup T))-\lambda_i(\boldsymbol{L}(S)))(\lambda_i(\boldsymbol{L}(S\cup T))+\lambda_i(\boldsymbol{L}(S)))}{\lambda_i^2(\boldsymbol{L}(S))\lambda_i^2(\boldsymbol{L}(S\cup T))} \\
&\geq n\frac{(\mathrm{Tr}\left(\boldsymbol{L}(S\cup T)\right)-\mathrm{Tr}\left(\boldsymbol{L}(S)\right))(\lambda_2(\boldsymbol{L}(S\cup T))+\lambda_2(\boldsymbol{L}(S))}{\lambda_n^2(\boldsymbol{L}(S))\lambda_n^2(\boldsymbol{L}(S\cup T))} \\
&= \frac{2n|T\backslash S|(\lambda_2(\boldsymbol{L}(S\cup T))+\lambda_2(\boldsymbol{L}(S))}{\lambda_n^2(\boldsymbol{L}(S))\lambda_n^2(\boldsymbol{L}(S\cup T))}.\\ 
\end{aligned}
$$

On the other hand,
$$
\begin{aligned}
\Theta_T(S)&=n\sum\nolimits_{i=2}^n\frac{(\lambda_i(\boldsymbol{L}(S\cup T))-\lambda_i(\boldsymbol{L}(S)))(\lambda_i(\boldsymbol{L}(S\cup T))+\lambda_i(\boldsymbol{L}(S)))}{\lambda_i^2(\boldsymbol{L}(S))\lambda_i^2(\boldsymbol{L}(S\cup T))} \\
&\leq n\frac{(\mathrm{Tr}\left(\boldsymbol{L}(S\cup T)\right)-\mathrm{Tr}\left(\boldsymbol{L}(S)\right))(\lambda_n(\boldsymbol{L}(S\cup T))+\lambda_n(\boldsymbol{L}(S))}{\lambda_2^2(\boldsymbol{L}(S))\lambda_2^2(\boldsymbol{L}(S\cup T))}\\
&=\frac{2n|T\backslash S|(\lambda_n(\boldsymbol{L}(S\cup T))+\lambda_n(\boldsymbol{L}(S))}{\lambda_2^2(\boldsymbol{L}(S))\lambda_2^2(\boldsymbol{L}(S\cup T))}.
\end{aligned}
$$
Then, we put the above two bounds together and derive the lower bound of the submodular ratio $\gamma.$
$$
\begin{aligned}
&\quad \frac{\sum\nolimits_{e\in T\setminus S}\Theta_{e}(S)}{\Theta_T(S)} \\
&\geq\sum\nolimits_{e\in T\setminus S}\frac{(\lambda_2(\boldsymbol{L}(S\cup {e}))+\lambda_2(\boldsymbol{L}(S))}{\lambda_n^2(\boldsymbol{L}(S))\lambda_n^2(\boldsymbol{L}(S\cup e))}\cdot \frac{\lambda_2^2(\boldsymbol{L}(S))\lambda_2^2(\boldsymbol{L}(S\cup T))}{(\lambda_n(\boldsymbol{L}(S\cup T))+\lambda_n(\boldsymbol{L}(S))} \\
& \geq\left(\frac{\lambda_2(\boldsymbol{L})}{\lambda_n(\boldsymbol{L}(Q))}\right)^5=\left(\frac{\lambda_2(\boldsymbol{L})}n\right)^5.
\end{aligned}
$$
The last equality holds since the largest eigenvalue of the Laplacian matrix of a complete graph with $n$ nodes is $n.$
Similarly, we derive the upper bound of the curvature $\alpha.$ Let $j$ be any candidate edge in $S \backslash T.$ Then, we have
$$
\begin{aligned}
&\quad  \frac{\Theta_{j}(S\cup T\backslash j)}{\Theta_{j}(S\backslash j)} \\
&\geq\frac{2n(\lambda_2(\boldsymbol{L}(S\cup T\backslash j))+\lambda_2(\boldsymbol{L}(S\cup T))}{\lambda_n^2(\boldsymbol{L}(S\cup T\backslash j))\lambda_n^2(\boldsymbol{L}(S\cup T))}\cdot\frac{\lambda_2^2(\boldsymbol{L}(S\backslash j))\lambda_2^2(\boldsymbol{L}(S))}{2n(\lambda_n(\boldsymbol{L}(S \backslash j))+\lambda_n(\boldsymbol{L}(S))} \\
& \geq\left(\frac{\lambda_2(\boldsymbol{L})}{\lambda_n(\boldsymbol{L}(Q))}\right)^5=\left(\frac{\lambda_2(\boldsymbol{L})}n\right)^5,
\end{aligned}
$$
which combing with Definition~\ref{curvature} completes the proof.
\end{proof}

\subsection{Simple Greedy Approach}
We can solve Problem~\ref{problem} by the following na\"ive brute-force approach. For each subset of edges $T$ of the $\binom{|Q|}{k}$ possible subsets of edges, we can compute the total biharmonic distance for its associated graph when all edges in this set are added. Then, output the subset $T^*$ of edges, whose addition leads to the largest decrease in the total biharmonic distance. Although this method is simple, it is computationally impossible even for small networks since computing the total biharmonic distance involves inverting a matrix, a process that requires $\Omega(n^3)$ time. This results in a total complexity of $\Omega(\binom{|Q|}{k}n^3)$, which is very expensive even for small values of $k$.

To address the exponential complexity, we employ greedy heuristics. Based on Lemma~\ref{marginal}, we can devise a simple greedy approach choosing the largest marginal decrease in each iteration, as formulated in Algorithm~\ref{alg:deter}. 


\begin{algorithm}[h]	
\caption{\textsc{DeterDiff}$(G,k)$}\label{alg:deter}
\Input{Graph $G=(V,E)$; an integer $k$ }
\Output{An edge set $T \subseteq Q$ of size $k$}
            Compute $\LL^\dagger$, $\LL^{2\dagger}$, $\LL^{3\dagger}$ \;
		$T = \emptyset$,  $Q=(V\times V)\backslash E$ \; 
		\For{$i = 1$ to $k$}{
                Compute $ \Delta(e)=B(G)-B({G_{\{e\}}})$ for each $e \in Q$\;
                Select $e_i$ s.t. $e_i \gets \mathrm{arg\, max}_{e \in Q} \Delta(e)$ \;
			$T \gets T \cup \{e_i\}$, $G \gets G(V,E \cup \{e_i\})$, $Q \gets Q \backslash \{e_i\}$ \;
                Update $\LL^\dagger$ by ~\eqref{morrison}, $\LL^{2\dagger} $ by ~\eqref{Lsquare}, and $ \LL^{3\dagger} $ by \eqref{Lcube} \;
		}
		\Return $T$
\end{algorithm}

We maintain three matrices: \(\LL^\dagger\), \(\LL^{2\dagger}\), and \(\LL^{3\dagger}\). These matrices are precomputed with a time complexity of \(O(n^3)\) (Line 1). In each iteration, calculating \(\Delta(e)\) for any edge now takes \(O(1)\) per edge. Since \(\Delta(e)\) needs to be evaluated for all edges in \(Q\) to find the one with the maximum marginal decrease, the total cost per iteration is \(\Omega(n^2)\). Upon selecting and adding an edge \(e\) to the graph, \(\LL^\dagger\), \(\LL^{2\dagger}\), and \(\LL^{3\dagger}\) are updated using the Sherman-Morrison formula, taking \(O(n^2)\) for each update rather than directly inverting a matrix in time $O(n^3)$. Specifically, 
$\LL^\dagger$ is updated using~\eqref{morrison},
\(\LL^{2\dagger}\) is updated using~\eqref{Lsquare}, while \(\LL^{3\dagger}\) is updated according to the following equation:
\begin{equation}\label{Lcube}
\begin{aligned}
(\LL + \bb_e \bb_e^T)^{3\dag}=& \LL^{3\dagger}+ \frac{3\boldsymbol{L}^{2\dagger} \boldsymbol{b}_e\boldsymbol{b}_e^\top\boldsymbol{L}^{2\dagger}\boldsymbol{b}_e\boldsymbol{b}_e^\top \boldsymbol{L}^{\dagger}}{(1+\boldsymbol{b}_e^\top\boldsymbol{L}^\dagger\boldsymbol{b}_e)^2} - \frac{3\boldsymbol{L}^{2\dagger} \boldsymbol{b}_e\boldsymbol{b}_e^\top\boldsymbol{L}^{2\dagger}}{1+\boldsymbol{b}_e^\top\boldsymbol{L}^\dagger\boldsymbol{b}_e} -
\\&\frac{\boldsymbol{L}^{\dagger} \boldsymbol{b}_e\boldsymbol{b}_e^\top\boldsymbol{L}^{2\dagger}\boldsymbol{b}_e\boldsymbol{b}_e^\top\boldsymbol{L}^{2\dagger}\boldsymbol{b}_e\boldsymbol{b}_e^\top \boldsymbol{L}^{\dagger}}{(1+\boldsymbol{b}_e^\top\boldsymbol{L}^\dagger\boldsymbol{b}_e)^3}. 
\end{aligned}
\end{equation}
Thus, Algorithm~\ref{alg:deter} has a time complexity of \(O(n^3 + kn^2)\), much faster than the brute-force algorithm. We then analyze its approximation guarantee using the following theorem.

\begin{theorem}[\cite{BiBuKrTs17}]\label{bound guarantee}
The greedy algorithm offers an approximation ratio of at least \(\frac{1}{\alpha}\left(1 - e^{-\alpha\gamma}\right)\) for the problem of maximizing a non-negative increasing set function \(f\) with submodularity ratio \(\gamma\) and curvature \(\alpha\).
\end{theorem}

Theorem~\ref{bound guarantee}, together with the bounds for the submodularity ratio \(\gamma\) and curvature \(\alpha\) of the function \(f(T) = B(G) - B(G_T)\) stated in Lemma~\ref{bound}, leads to a performance analysis for Algorithm~\ref{alg:deter}. To enhance the theoretical accuracy of Algorithm~\ref{alg:deter}, a larger value of \((1 - e^{- \gamma \alpha}) / \alpha\) is desirable, which can be achieved with a higher \(\gamma\) and a lower \(\alpha\). According to Lemma~\ref{bound}, the bounds for \(\gamma\) and \(\alpha\) are closely related to \(\lambda_2(\boldsymbol{L})\).   A larger \(\lambda_2(\boldsymbol{L})\) results in a tighter theoretical bound. Since the complete graph has \(\lambda_2(\LL) = n\), denser graphs are likely to yield a better approximation ratio.  While the theoretical bound implies that the approximation relies heavily on the graph's structure, Algorithm~\ref{alg:deter} often performs very close to the optimal solutions in practice, as shown in Section~\ref{experimen}.

\subsection{Gradient-Based Greedy Approach}
Traditional heuristic greedy algorithms select edges based on the maximum marginal decrease or gain in the objective function. While this approach is commonly employed, researchers have also explored using gradients or partial derivatives of the relevant function to assess the impact of edge modifications~\cite{LiZh18,black2023understanding}. In this subsection, we explore the gradient of the total biharmonic distance to characterize the significance of an edge instead of the marginal decrease $\Delta(e)$. 

\begin{lemma}
\label{lemma:gradient}
For a connected graph \( G = (V, E) \), the gradient of total biharmonic distance \( c(e) \) for an edge \( e \in (V \times V) \backslash E \) is \(2n \boldsymbol{b}_e^\top \boldsymbol{L}^{3\dagger} \boldsymbol{b}_e \).
\end{lemma}
\begin{proof}Let $w_e$ denote the weight of edge $e$. We have
$$
\begin{aligned}
c(e)&=\frac{\partial B(G_{\{e\}})}{\partial w_e}=n\frac{\partial \trace{(\boldsymbol{L}+\bb_e\bb_e^T)^{2\dagger})}}{\partial w_e}\\
&=n\lim_{\epsilon\to0}\frac1\epsilon\left[\trace{(\boldsymbol{L}+\epsilon\bb_e\bb_e^T)^{2\dagger}}-\trace{\boldsymbol{L}^{2\dagger}}\right].
\end{aligned}
$$
By~\eqref{Lsquare}, we obtain
$$
\begin{aligned}
&\trace{(\boldsymbol{L}+\epsilon\bb_e\bb_e^T)^{2\dagger}}-\trace{\LL^{2\dag}} \\
=& \epsilon \mathrm{Tr}\left( \frac{\LL^{2\dag} \bb_e \bb_e^T \LL^{\dag}}{1 + \epsilon \bb_e^T \LL^{\dag} \bb_e}\right)+ \epsilon \mathrm{Tr}\left(\frac{\LL^{\dag} \bb_e \bb_e^T \LL^{2\dagger}}{1 + \epsilon \bb_e^T \LL^{\dag} \bb_e}\right) - \mathrm{Tr}\left(\frac{(\epsilon\LL^{\dag} \bb_e \bb_e^T \LL^{\dag})^2}{(1 + \epsilon \bb_e^T \LL^{\dag} \bb_e)^2}\right)\\
=& \frac{2\epsilon\bb_e^T\LL^{3\dag} \bb_e}{1 + \epsilon \bb_e^T \LL^{\dag} \bb_e}- \frac{(\epsilon  \bb_e^T\LL^{\dag} \bb_e)^2}{(1 + \epsilon \bb_e^T \LL^{\dag} \bb_e)^2}.
\end{aligned}
$$
Then, we can derive that
$$
\begin{aligned}
c(e)
&=n\lim_{\epsilon\to0}\frac1\epsilon\left[\frac{2\epsilon\bb_e^T\LL^{3\dag} \bb_e}{1 + \epsilon \bb_e^T \LL^{\dag} \bb_e}- \frac{(\epsilon  \bb_e^T\LL^{\dag} \bb_e)^2}{(1 + \epsilon \bb_e^T \LL^{\dag} \bb_e)^2}\right]=2n\boldsymbol{b}_e^\top\boldsymbol{L}^{3\dagger}\boldsymbol{b}_e.
\end{aligned}
$$
This completes the proof.\end{proof}


For each edge \( e \in (V \times V) \backslash E \), the gradient of the total biharmonic distance is given by \( 2n \boldsymbol{b}_e^\top \boldsymbol{L}^{3\dagger} \boldsymbol{b}_e \). Since \( n \) is constant, the critical term reduces to \( \boldsymbol{b}_e^\top \boldsymbol{L}^{3\dagger} \boldsymbol{b}_e \). We redefine \( c(e) \) as \( \boldsymbol{b}_e^\top \boldsymbol{L}^{3\dagger} \boldsymbol{b}_e \). This leads to a gradient-based greedy algorithm that iteratively selects the edge with the largest gradient, as described in Algorithm~\ref{alg:DeterGrad}. The time complexity of this algorithm is \( O(n^3 + kn^2) \), following a similar analysis to that of Algorithm~\ref{alg:deter}. One interesting finding is that $c(e)$ is exactly the 3-harmonic distance. Recalling that the gradient of the Kirchhoff index is proportional to the biharmonic distance~\cite{YiShLiZh18}, we can deduce that the gradient of the total $k$-harmonic distance is proportional to the ($k+1$)-harmonic distance by a similar proof of Lemma~\ref{lemma:gradient}.

While the strategies for selecting edges in the greedy process differ between Algorithm~\ref{alg:deter} and Algorithm~\ref{alg:DeterGrad}, both approaches yield relatively effective solutions in the experimental evaluations presented in Section~\ref{experimen}. This indicates that using the gradient instead of the marginal decrease is a feasible method. This method is very common and performs well in our problem. In the next section, based on Algorithm \textsc{DeterGrad}, we design a nearly linear time algorithm to solve Problem~\ref{problem}.

\begin{algorithm}	
\caption{\textsc{DeterGrad}$(G,k)$}
\label{alg:DeterGrad}
\Input{Graph $G=(V,E)$; an integer $k$ }
\Output{An edge set $T \subseteq Q$ of size $k$}
             Compute $\LL^\dagger$, $\LL^{2\dagger}$, $\LL^{3\dagger}$ \;
		 $T = \emptyset$, $Q=(V\times V)\backslash E$ \;
		\For{$i = 1$ to $k$}{
                Compute $ c(e)$ for each $e \in Q$\;
                Select $e_i$ s.t. $e_i \gets \mathrm{arg\, max}_{e \in Q} c(e)$ \;
			 $T \gets T \cup \{e_i\}$, $G \gets G(V,E \cup \{e_i\})$, $Q \gets Q \backslash \{e_i\}$\;
                Update $\LL^\dagger$ by~\eqref{morrison}, $\LL^{2\dagger} $ by~\eqref{Lsquare}, and $\LL^{3\dagger} $ by~\eqref{Lcube} \;

		}
        
		\Return $T$
\end{algorithm}

\section{Fast Approximation Algorithm}

Both Algorithm~\ref{alg:deter} and Algorithm \textsc{DeterGrad} struggle with high time complexity, making them computationally prohibitive for large-scale networks. This challenge arises from several key factors. First, calculating \(c(e)\) requires inverting the Laplacian matrix, which has a time complexity of \( \Theta(n^3) \), presenting a significant computational burden. Second, each iteration involves examining all pairs of nodes to identify an edge to add and requires querying \( |Q| \) gradients, where \( n^2 - m - k \leq |Q| \leq n^2 - m \). This leads to a computational bottleneck. Finally, even with the use of the Sherman-Morrison formula to update \(\LL^{\dagger}\), \(\LL^{2\dagger}\), and \(\LL^{3\dagger}\), the time complexity per iteration remains \(O(n^2)\). Our primary objective is to optimize three aspects for improved efficiency. We begin by addressing the highest time cost, which is our first challenge. To reduce the high computational cost associated with matrix inversion, we employ the projection technique and the Laplacian solver to estimate \(c(e)\). This approach also eliminates the need to update \(\LL^{\dagger}\), \(\LL^{2\dagger}\), and \(\LL^{3\dagger}\), effectively addressing the third challenge simultaneously. For the second challenge, we introduce a pruning strategy based on geometric principles to approximate the convex hull, which helps to reduce the size of the candidate edge set. By integrating these methods, we propose an approximation algorithm, outlined in Algorithm~\ref{ApproxFast}. This algorithm achieves an \(\epsilon\)-approximation of \(c(e)\) in each of the \(k\) iterations, enabling approximate solutions with nearly linear time complexity.

\subsection{Projection Technique and Laplacian Solver}\label{projectionsolver}
We begin by reformulating \( \boldsymbol{b}_e^\top \boldsymbol{L}^{3\dagger} \boldsymbol{b}_e \) in terms of the \(\ell_2\)-norm.

$$
\begin{aligned}
\boldsymbol{b}_e^\top\boldsymbol{L}^{3\dagger}\boldsymbol{b}_e &= (\LL^\dagger \boldsymbol{b}_e)^\top \LL^\dagger \LL^\dagger \boldsymbol{b}_e = (\LL^\dagger \boldsymbol{b}_e)^\top \LL^\dagger \LL \LL^\dagger \LL^\dagger \boldsymbol{b}_e\\
&= (\LL^\dagger \boldsymbol{b}_e)^\top \LL^\dagger \BB^\top \BB \LL^\dagger \LL^\dagger \boldsymbol{b}_e = \|\BB \LL^{2\dagger}(\ee_u-\ee_v)\|^2.
\end{aligned}
$$

In this way, computing \( c(e) \) reduces to computing \( \|\BB \LL^{2\dagger}(\ee_u-\ee_v)\|^2 \) in \(\mathbb{R}^m\). However, calculating this exact \(\ell_2\)-norm is computationally intensive because the dimension \( m \) of the vectors \( \BB \LL^{2\dagger} \ee_i\) (where \( i = 1, 2, \ldots, n \)) is high, and it also requires matrix inversion. To alleviate this issue and reduce the dimensionality of the vectors, we employ the Johnson-Lindenstrauss lemma for dimensionality reduction.

\begin{lemma}[\cite{JoLi84,Ac01}] \label{Jllemma}
Given fixed vectors \(\vvv_1,\vvv_2,\ldots,\vvv_n\in \mathbb{R}^d\) and \(\beta>0\), let \(\QQ_{t\times d}\) be a random matrix, where each entry is either \(1/\sqrt{t}\) or \(-1/\sqrt{t}\), with probability of 1/2, and \(t \geq 24\log (n)/\beta^2\). Then, with probability at least \(1-1/n\), the following holds for all pairs \(i,j \leq n\):
\[
(1-\beta)\|\vvv_i-\vvv_j\|^2_2 \leq \|\QQ \vvv_i - \QQ \vvv_j\|^2_2 \leq (1+\beta)\|\vvv_i-\vvv_j\|^2_2.
\]
\end{lemma}
Let \(\boldsymbol{Q} \in \mathbb{R}^{t \times m}\) be a random matrix with entries \(\pm 1/\sqrt{t}\), where \(t = \lceil24\log(n)/\beta^2\rceil\). By Lemma~\ref{Jllemma}, we have the approximation:
\begin{equation}\label{eqimportant}
 \|\boldsymbol{Q}\BB \LL^{2\dagger}(\ee_u-\ee_v)\|^2  \approx_{\beta} \|\BB \LL^{2\dagger}(\ee_u-\ee_v)\|^2
\end{equation}
for every pair \((u,v)\) with probability at least \(1 - 1/n\).

We have projected a set of \(n\) \(m\)-dimensional vectors onto a low \(t\)-dimensional subspace spanned by the columns of the matrix \(\QQ\). However, the high computational cost of \(\Omega(n^3)\) associated with matrix inversion still exists. To circumvent this expensive operation, we utilize a fast SDD linear system solver to solve the corresponding linear equations.

\begin{lemma}[\cite{SpTe14,CoKyMiPaJaPeRaXu14}]\label{solver}
Let \(\XX \in\mathbb{R}^{n\times n}\) be a positive semi-definite matrix with \(m\) nonzero entries, and \(\bb \in\mathbb{R}^n\) a vector. For any error parameter \(\delta>0\), there exists a solver, denoted by $\aa = \SDDMSolver(\XX, \bb, \delta)$, which returns a vector \(\aa \in \mathbb{R}^n\) satisfying
\[
\left \| \aa - \XX^{-1} \bb \right \|_{\XX} \leq \delta \left \| \XX^{-1} \bb \right \|_{\XX}
\]
with probability at least \(1 - 1/n\). The expected runtime of this solver is \(\tilde{O}(m)\), where \(\tilde{O}(\cdot)\) suppresses poly(\(\log (n)\)) factors.
\end{lemma}

To approximate \(\|\BB \LL^{2\dagger}(\ee_u-\ee_v)\|^2\), we apply the solver as follows. For any small \(\beta > 0\), we choose an appropriate \(\delta\) to ensure the approximation error remains within \(\beta\) with probability at least \(1 - 1/n\). Let \(\ZZ = \QQ\BB\LL^{2\dagger}\). We solve the system of equations \(\LL^{2} \zz_i = \yy_i\), with the constraint \(\one^\top \zz_i = 0\), for \(i = 1, \dots, t\), where \(\zz^\top_i\) and \(\yy^\top_i\) represent the \(i\)-th rows of \(\ZZ\) and \(\QQ\BB\), respectively. By using the solver \(\SDDMSolver(\LL, \yy_i, \delta)\) twice, and carefully selecting \(\delta\), we achieve a reliable approximation of \(\zz_i\). Before proving this, we first introduce some lemmas (all missing proofs are provided in Appendix~\ref{appendix}).

\begin{lemma}\label{lem:Lnorm}
	Given a Laplacian matrix $\LL$ and a vector $\zz$ orthogonal to $\one$, we have
 \begin{equation}\label{eq:Lnorm}
		\norm{\LL^{\dagger}\zz}_{\LL} \le 2n^4 \norm{\zz}_{\LL}.
	\end{equation}
\end{lemma}

\begin{proof}$\norm{\LL^{\dagger}\zz}_{\LL} = \sqrt{\zz^\top\LL^{\dagger}\zz} \le \sqrt{2n^4}\norm{\zz},$
and $\norm{\zz}_{\LL}  = \sqrt{\zz^\top\LL\zz}   \ge \frac{1}{\sqrt{2n^4}}\norm{\zz},$ which further suggest that $ \norm{\LL^{\dagger}\zz}_{\LL} \le 2n^4 \norm{\zz}_{\LL}.$
\end{proof}

\begin{lemma}~\label{2laplacian}
	Given two vector sequences \(\zz^{(1)},\zz^{(2)}\) and \(\zztil^{(1)},\zztil^{(2)}\) such that
	\[
		\zz^{(1)} = \LL^{\dagger} \yy, \quad \zz^{(2)} = \LL^{\dagger} \zz^{(1)},
	\]
	and
	\[
		\zztil^{(1)} = \SDDMSolver(\LL, \yy, \delta), \quad \zztil^{(2)} = \SDDMSolver(\LL, \zztil^{(1)}, \delta),
	\]
we have the following error bound:
	\begin{equation}\label{eq:errorZ}
	  		\norm{\zztil^{(2)} - \zz^{(2)}}_{\LL} \leq \delta' \norm{\LL^{\dagger}\yy}_{\LL},  
	\end{equation}
	where \(\delta' = 2n^4 \left( \delta^2 + 2\delta \right)\).
\end{lemma}
\begin{proof}First, according to Lemma~\ref{solver}, we have
	\begin{equation}\label{eq:lap}
		\norm{\zztil^{(1)}-\LL^{\dagger}\yy}_{\LL} \le \delta \norm{\LL^{\dagger}\yy}_{\LL},
	\end{equation}
 	\begin{equation}\label{eq:lap2}
		\norm{\zztil^{(2)}-\LL^{\dagger}\zztil^{(1)}}_{\LL} \le \delta \norm{\LL^{\dagger}\zztil^{(1)}}_{\LL}.
	\end{equation}
 Since $\norm{\cdot}_{\LL}$ is a norm, (\ref{eq:lap}) can be recast as
	\begin{equation}\label{recast}
	(1-\delta) \norm{\LL^{\dagger}\yy}_{\LL} \le \norm{\zztil^{(1)}}_{\LL} \le (1+\delta) \norm{\LL^{\dagger}\yy}_{\LL}.
	\end{equation}
Then, we apply the triangle inequality on the left-hand side of (\ref{eq:errorZ}):
	\begin{align*}
		\norm{\zztil^{(2)}-\zz^{(2)}}_{\LL} 
		\le& \norm{\zztil^{(2)} - \LL^{\dagger}\zztil^{(1)}}_{\LL} + \norm{\LL^{\dagger}\zztil^{(1)}-\LL^{\dagger}\zz^{(1)}}_{\LL}\\
		\le& \delta \norm{\LL^{\dagger}\zztil^{(1)}}_{\LL} + 2n^4 \norm{\zztil^{(1)}-\zz^{(1)}}_{\LL} \\
  \le& \delta (2n^4)\norm{\zztil^{(1)}}_{\LL} + 2n^4 \delta \norm{\LL^{\dagger}\yy}_{\LL}\\
    \le& \delta (2n^4)(1+\delta) \norm{\LL^{\dagger}\yy}_{\LL} + 2n^4 \delta \norm{\LL^{\dagger}\yy}_{\LL}\\
		\le& 2n^4 \kh{\delta^2+2\delta}\norm{\LL^{\dagger}\yy}_{\LL},
	\end{align*}
	where the second inequality is due to (\ref{eq:Lnorm}) and (\ref{eq:lap}), the third inequality is due to (\ref{eq:Lnorm}) and the fourth inequality is due to (\ref{recast}).\end{proof}

 We then show the constructed $\| \ZZtil(\ee_u - \ee_v) \|^2$ is a reliable approximation of $\|\BB \LL^{2\dagger}(\ee_u-\ee_v)\|^2$ as stated in the following lemma.
\begin{lemma}\label{lem:error}
Given a \(t \times n\) matrix \(\ZZ\), for any $0 \le \beta \leq 1/3$, suppose that for any pair of nodes \(u,v \in V\):
\[
	(1-\beta) \|\BB \LL^{2\dagger} (\ee_u - \ee_v)\|^2 \leq \|\ZZ (\ee_u - \ee_v)\|^2 \leq (1+\beta) \|\BB \LL^{2\dagger} (\ee_u - \ee_v)\|^2.
\]
Let \(\zz^\top_i\) be the \(i\)-th row of \(\ZZ\), and \(\tilde{\zz}_i\) be an approximation of \(\zz_i\), satisfying
\begin{equation}\label{laplacain eq}
    \|\zz_i - \tilde{\zz}_i\|_{\LL} \leq \delta' \|\LL^{\dagger} \yy_i\|_{\LL},
\end{equation}
where
\begin{equation}\label{delta'}
      \delta' \leq \frac{\beta}{3n^2} \sqrt{\frac{2(1-\beta)}{(1+\beta)(n-1)}}.
\end{equation}
Then, for any node pair \(u,v \in V\), the following relation holds
\begin{equation}\label{eqimportant1}
 \| \ZZtil(\ee_u - \ee_v) \|^2 \approx_{\frac{7}{3}\beta} \|\BB \LL^{2\dagger}(\ee_u-\ee_v)\|^2.
\end{equation}

\end{lemma}

To ensure Lemma~\ref{2laplacian} holds, we should carefully select \(\delta\) for the Laplacian solver. By substituting \(\delta' = 2n^4 (\delta^2 + 2\delta)\) from Lemma~\ref{2laplacian}, we obtain the inequality
$\delta' = 2n^4 (\delta^2 + 2\delta) \leq \frac{\beta}{3n^2} \sqrt{\frac{2(1-\beta)}{(1+\beta)(n-1)}}$. This allows us to refine \(\delta\) as 
$
\delta \leq \left( \frac{\beta}{6n^6} \sqrt{\frac{2(1-\beta)}{(n-1)(1+\beta)}} + 1 \right)^{1/2} - 1.
$

By setting \(\delta\) accordingly, we can have the approximation in Lemma~\ref{lem:error}. Although Lemma~\ref{solver} ensures that we can estimate \(c(e)\) in \(\tilde{O}(m/\beta^2)\) time, evaluating \(c(e)\) for all edges in the candidate edge set \(Q\) remains computationally demanding. To address this, we introduce a pruning technique to reduce the size of the candidate edge set.

\subsection{Pruning Technique}

In this subsection, we utilize a pruning technique based on the convex hull of a set of points~\cite{Ro70}. This approach allows us to reduce the candidate edge set \( Q \) to a smaller subset \( P \) with \( P \subset Q \) and \( |P| \ll |Q| \). This reduction improves the efficiency of our algorithm and provides approximate solutions for the optimization problem.

Based on \eqref{eqimportant1}, estimating \( c(e) \) can be reduced to calculating distances between \( O(n^2) \) pairs of points \(\tilde{\ZZ}\ee_i\) in \(\mathbb{R}^t\), where \( i = 1, 2, \ldots, n\). We denote a set \( P = \{p_1, p_2, \ldots, p_n\} \), where each \( p_i = \tilde{\ZZ}\ee_i \). To find the maximum gradient of \( c(e) \), we need to identify the farthest pair of points among \( P \). As distances between interior points are usually smaller than those involving boundary points, we concentrate on identifying the boundary points, which can be efficiently determined using the concept of the convex hull.

\begin{definition}[\cite{Ro70}]\label{convex}
Given a set of \( n \) points \( P = \{p_1, p_2, \ldots, p_n\} \) in \(\mathbb{R}^t\), the convex hull of \( P \) is the minimal convex polytope that contains all the points in \( P \). The boundary of this convex hull is denoted by \( C(P) \), and the subset of points in \( P \) that lie on this boundary is denoted by \(\bar{P}\).
\end{definition}

Below, we explain how to efficiently identify the convex hull of a set of points. For \( P = \{p_1, p_2, \ldots, p_n\} \subset \mathbb{R}^t \), the time complexity of finding the convex hull is \( O(n^{\lfloor t/2 \rfloor}) \) \cite{Ch93, MaOtMaMa08}, which is impractical, even when \( t = \lceil 24 \log(n)/\beta^2 \rceil \). As a result, convex hull approximations are often utilized. Among the available methods, \textsc{FastHull} \cite{AwKaZh18, AwKaZh20, Ka15} offers the lowest time complexity while still providing an error guarantee. It outperforms other approaches, such as linear programming \cite{ruano2015randomized} and the semi-nonnegative matrix factorization method with a kernel trick~\cite{huang2018kernelized}. We choose to use \textsc{FastHull} for approximating the convex hull due to its significantly lower time complexity \cite{AwKaZh20}. The performance of the \textsc{FastHull} algorithm is captured in the following lemma. Let \( d(P) \) denote the diameter of the set \( P \), defined as \( d(P) = \max_{p_i, p_j \in P} \left|\left|p_i - p_j\right|\right|_2 \). This represents the maximum distance between all pairs of points \( p_i \) and \( p_j \) and satisfies the relationship \( d(P) = d(\bar{P}) \).

\begin{lemma}[\cite{AwKaZh20}]\label{convex time}
For a set of points $P=\{p_1,p_2,\ldots,p_n\} \subset \mathbb{R}^t$, and a parameter $\mu\in (0,1)$, the algorithm $\textsc{FastHull}(P,\mu)$ generates an $l$-node subset $\hat{P}$ of $\bar{P}$. The algorithm has a time complexity of $O(nl(t+\mu^{-2}))$ and ensures that the Euclidean distance for any $p\in \bar{P}$ to $C(\hat{P})$ does not exceed $\mu d(P)$.
\end{lemma}

We can obtain an approximate point set \( \hat{P} \subseteq \bar{P} \), the boundary of convex hull $C(P)$ by applying \textsc{FastHull}. It is easy to derive that $d(\hat{P}) \geq (1-2\mu) d(P)$. Since \( (1 - 2\mu)^2 \geq 1 - 4\mu \), we have $
(1 - 4\mu) d(P)^2 \leq d(\hat{P})^2.$
Thus,
\begin{equation}\label{eq1}
 d(\hat{P})^2  \approx_{4\mu} d(P)^2  .
\end{equation}
According to Lemma~\ref{convex time} and \eqref{eq1}, \textsc{FastHull} can efficiently approximate the convex hull for \( P \). When \(t= \lceil24 \log(n)/\beta^2 \rceil \), the time complexity of \(\textsc{FastHull}\) is \(O(nl \log (n) \beta^{-2})\). This approximation greatly reduces computation time, as we only need to query \(l\) times, and finding the furthest point pairs from the set \( \bar{P} \) now only needs \( O(l^2 \log (n) \beta^{-2}) \). Based on the above results, we now propose a fast algorithm, \textsc{ApproxFast}, as detailed in Algorithm~\ref{ApproxFast}.{ The approximation ratio for each iteration is given by the following theorem.
\begin{theorem}
    The gradient of the edge selected at each iteration is an \(\epsilon\)-approximation of the maximum gradient of the objective function.
\end{theorem}

\begin{proof}
    By Lemma~\ref{lem:error}, we have
    $
    d(P)^2  \approx_{\frac{7}{3}\beta}  \max_{u,v \in V} \|\BB \LL^{2\dagger}(\ee_u - \ee_v)\|^2.
    $
    Combining this result with \eqref{eq1}, it follows that
   $
    d(\hat{P})^2 \approx_{\frac{7}{3}\beta + 4\mu}  \max_{u,v \in V} \|\BB \LL^{2\dagger}(\ee_u - \ee_v)\|^2.
    $
    Setting \(\beta = \frac{3}{14}\epsilon\) and \(\mu = \epsilon/8\), we obtain 
    $
    d(\hat{P})^2 \approx_{\epsilon}  \max_{e \in Q} c(e)$,
    which implies that the gradient of the edge selected in each iteration provides an \(\epsilon\)-approximation of the maximum gradient of the objective function.
\end{proof}

The time complexity of Algorithm~\ref{ApproxFast} is summarized in Theorem~\ref{performance}.
\begin{theorem}\label{performance}
    Given a positive integer $k$ and an error parameter $\epsilon \in(0,1)$, the runtime of Algorithm~\ref{ApproxFast} is in $\tilde{O}(k(nl+m)/\epsilon^2).$ 
\end{theorem}

The components, Laplacian solver and \textsc{FastHull}, of \textsc{ApproxFast} play a mutually complementary role in the algorithm. Removing either part would significantly increase the time complexity, making it difficult to scale the algorithm to large networks. If only the Laplacian solver is used, the complexity of each iteration becomes $O(n^2)$; if only FastHull is used, the complexity of each iteration becomes $O(n^3)$. So we did not conduct ablation studies in the subsequent experimental section.}

\begin{algorithm}[t]
	\caption{\textsc{ApproxFast}$(G,\epsilon,k)$}\label{ApproxFast}
		\Input{Graph $G=(V,E)$; a parameter $\epsilon$; an integer $k$ }
		\Output{ An edge set $T \subseteq Q$ of size $k$
		}
             $\mu = \epsilon/8$,  $\beta= \frac{3}{14}\epsilon$, $\delta = \left( \frac{\beta}{6n^6} \sqrt{\frac{2(1-\beta)}{(n-1)(1+\beta)}} + 1 \right)^{1/2} - 1 $\;
		 $T = \emptyset$,  $\LL \gets$ the Laplacian matrix of graph $G$ \;
            $t=\lceil 24\log(n)/\beta^2\rceil $ \;
           
            Generate random Gaussian matrices $\QQ \in \mathbb{R}^{t\times m}$\;
            \For{$i=1$ to $t$}{ 
            $\yy^{\top}_i $=the $i$-th row of $\QQ \BB$  \\}

		\For{$i = 1$ to $k$}{
                \For{$j=1$ to $t$}{
              $\zztil^{(1)}_j=\SDDMSolver(\LL, \yy_j, \delta)$\\
              $\zztil^{(2)}_j=\SDDMSolver(\LL, \zztil^{(1)}_j, \delta)$\\
	       $\tilde{\zz}_j=\zztil^{(2)}_j$\\
        }
                ${P}=\{{p}_1,{p}_2,\ldots,{p}_n\}$ with each point being $\tilde{\ZZ}_{:,1},\tilde{\ZZ}_{:,2},\ldots,\tilde{\ZZ}_{:,n}$\;
                $\hat{{P}}=\textsc{FastHull}({P},\mu)$ \;
                $(x,y)\leftarrow\arg\max_{u,v\in \hat{{P}}}\|\tilde{\ZZ}_{:,u}-\tilde{\ZZ}_{:,v}\|_2^2$\;
             $T\leftarrow T\cup\{(x,y)\}$, $G\leftarrow G(V,E\cup\{(x,y)\})$\;
                Update $\LL$\;
		}
		\Return $T$
\end{algorithm}

\section{Experiments}\label{experimen}
\subsection{Experimental Setup}
\textbf{Datasets. } We use 10 real-life networks from the Network Repository~\cite{RoAh15},  SNAP~\cite{LeSo16} and KONECT~\cite{kunegis2013konect}. Our experiments focus on the largest connected component (LCC) of these networks.

\textbf{Environment and Implementation.} All our experiments are conducted on a Linux server with a 2.10 GHz CPU and 128 GB of memory, with a single thread. We set $\epsilon=0.4$. All the proposed algorithms are implemented in \textit{Julia}, where the Laplace Solver is used from~\cite{KySa16}. 

 \textbf{Total Biharmonic Distance Computation.} For small graphs, we are able to compute the exact total biharmonic distance via computing $\LL^\dagger$. {For medium and large networks, it is difficult to obtain the exact total biharmonic distance. So we do not perform an analysis of the actual approximation error. Instead, we compare the values of the total biharmonic distance returned by different algorithms—smaller values indicate smaller approximation errors.} To efficiently compute the total biharmonic distance, we utilize Hutchinson's Monte-Carlo method~\cite{Hut89} and the Laplacian solver~\cite{SpTe14,CoKyMiPaJaPeRaXu14} to approximate the total biharmonic distance with a provable guarantee, adhering to the methodology in~\cite{ShYiZh18}. 

\textbf{Edge Addition Strategies.} The set of $k$ edges is selected by different strategies: \textsc{Optimal}, \textsc{Random}, \textsc{ACGreedy}, \textsc{DeterDiff}, \textsc{DeterGrad}, and \textsc{ApproxFast}. \textsc{Optimal} selects the $k$ edges that minimize $B(G)$ via brute-force search. As we are the first to solve Problem~\ref{problem}, and no existing methods are available for direct comparison, we include several heuristic baselines. \textsc{Random} selects $k$ edges uniformly at random. {\textsc{ACGreedy} selects $k$ edges that yield the greatest increase in algebraic connectivity~\cite{GhBo06}, defined as the second smallest eigenvalue of $\LL$. This method is based on the Fiedler vector, which is computed efficiently via the Lanczos algorithm~\cite{demmel1997applied}. We choose this method as a baseline because $B(G)$ is primarily influenced by the small eigenvalues of $\LL$, as shown in Definition~\ref{def:total biharmonic distance}.} In subsequent figures, we use $B(G)/n$ instead of $B(G)$.

{Details including dataset statistics, and parameter sensitivity analysis are provided in Appendix~\ref{appendix:expdetails}.}

\subsection{Effectiveness Comparison}

We begin by evaluating the effectiveness of our algorithms by comparing them with \textsc{Optimum}, \textsc{Random} and \textsc{ACGreedy}. To this end, we conduct experiments on four small networks: \textsf{Tribes} (16 nodes, 58 edges), \textsf{Karate} (34 nodes, 78 edges), \textsf{Dolphins} (62 nodes, 159 edges), and \textsf{Moreno} (64 nodes, 243 edges). Since these networks are small, we are able to obtain the optimal solution through exhaustive search. For each network, we select $k = 1, 2, \dots, 5$ edges. Figure~\ref{fig:2} illustrates how $B(G)/n$ changes as $k$ increases for each algorithm. 

The following observations can be made. First, the values returned by three greedy algorithms are almost identical to \textsc{Optimum}, with the four curves overlapping, indicating that \textsc{DeterDiff} significantly outperforms its theoretical guarantee. \textsc{DeterGrad} and \textsc{ApproxFast} demonstrate effective performance in practice. Second, our proposed algorithms all perform better than \textsc{Random} and \textsc{ACGreedy}. To further demonstrate the effectiveness of our algorithms, we compare their performance with \textsc{Random} and \textsc{ACGreedy} on four moderately larger real-world networks, each with fewer than 10,000 nodes. These networks are large enough to make \textsc{Optimum} impractical. For each network, Figure~\ref{exp1} shows the performance of each algorithm as $k$ increases from 1 to 10. {
All of our proposed algorithms perform significantly better than \textsc{ACGreedy}.} Notably, \textsc{DeterGrad} achieves results that are comparable to the standard greedy algorithm \textsc{DeterDiff}, suggesting that selecting the edge with the largest gradient is a highly effective strategy for Problem~\ref{problem}. Additionally, our fast approximation algorithm, \textsc{ApproxFast}, yields similar results, further demonstrating its effectiveness.

\begin{figure}[t]
		\centering
		\includegraphics[width=1\linewidth]{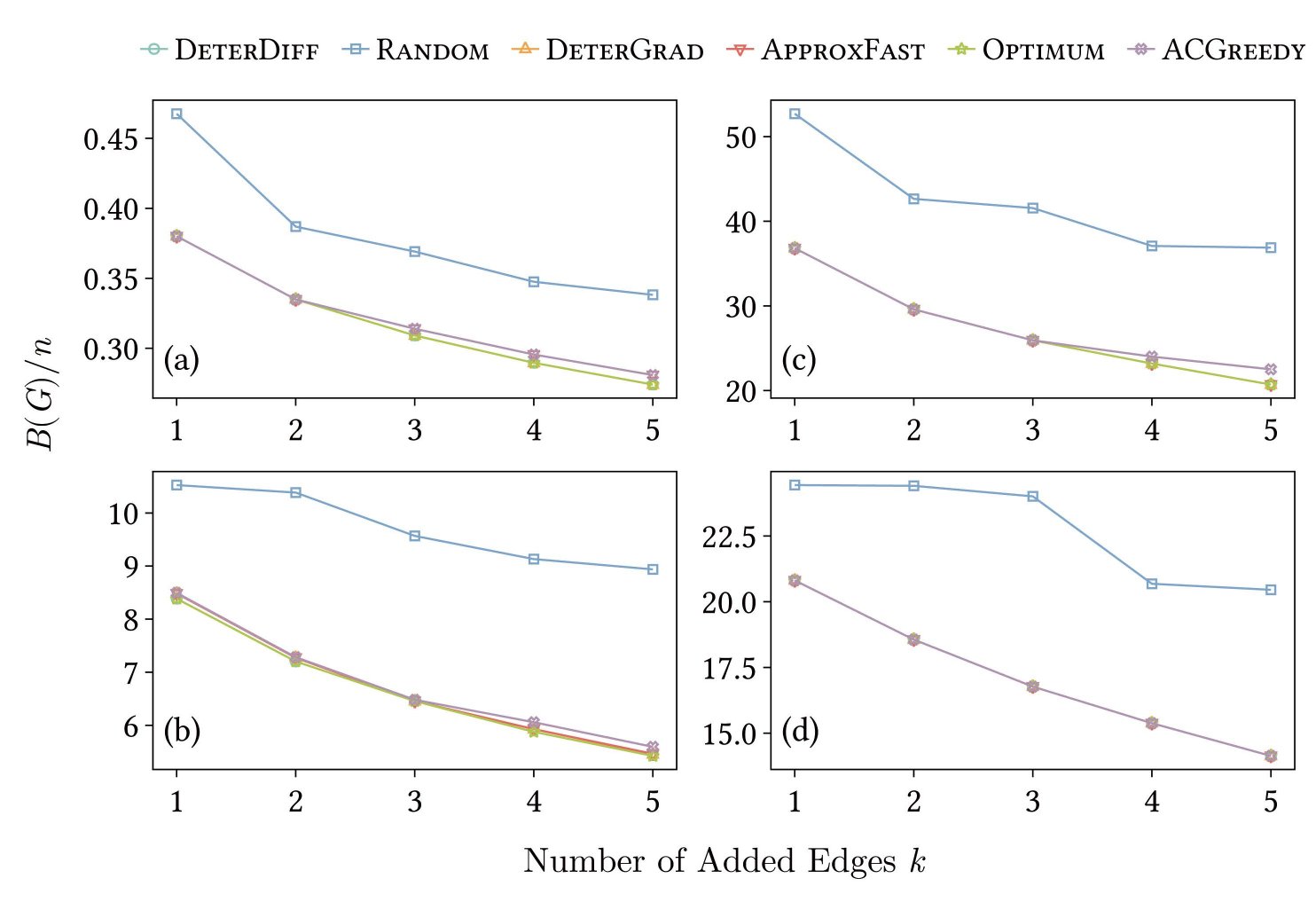}
		\caption{$B(G)/n$ returned by six algorithms on four small networks: \textsf{Tribes} (a), \textsf{Karate} (b), \textsf{Dolphins} (c), and \textsf{Moreno} (d).
        \label{fig:2}}
\end{figure}

\begin{figure}[t]
		\centering
		\includegraphics[width=1\linewidth]{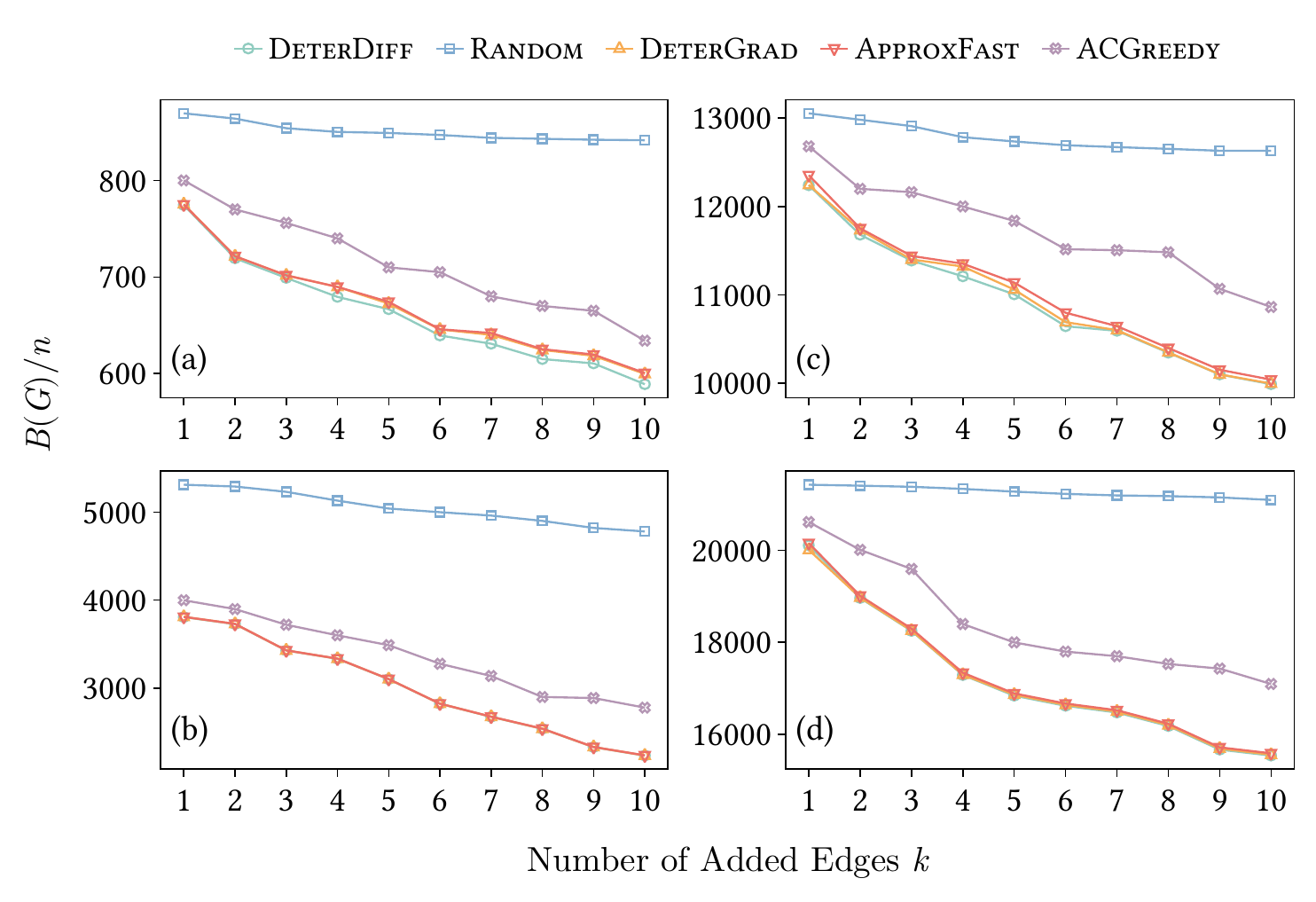}
		\caption{$B(G)/n$ returned by five algorithms on four networks:
  \textsf{Hamester} (a), \textsf{Facebook} (b), \textsf{Grqc} (c), \textsf{Hepth} (d). \label{exp1}}
\end{figure}

\begin{figure}[t]
		\centering
		\includegraphics[width=1\linewidth]{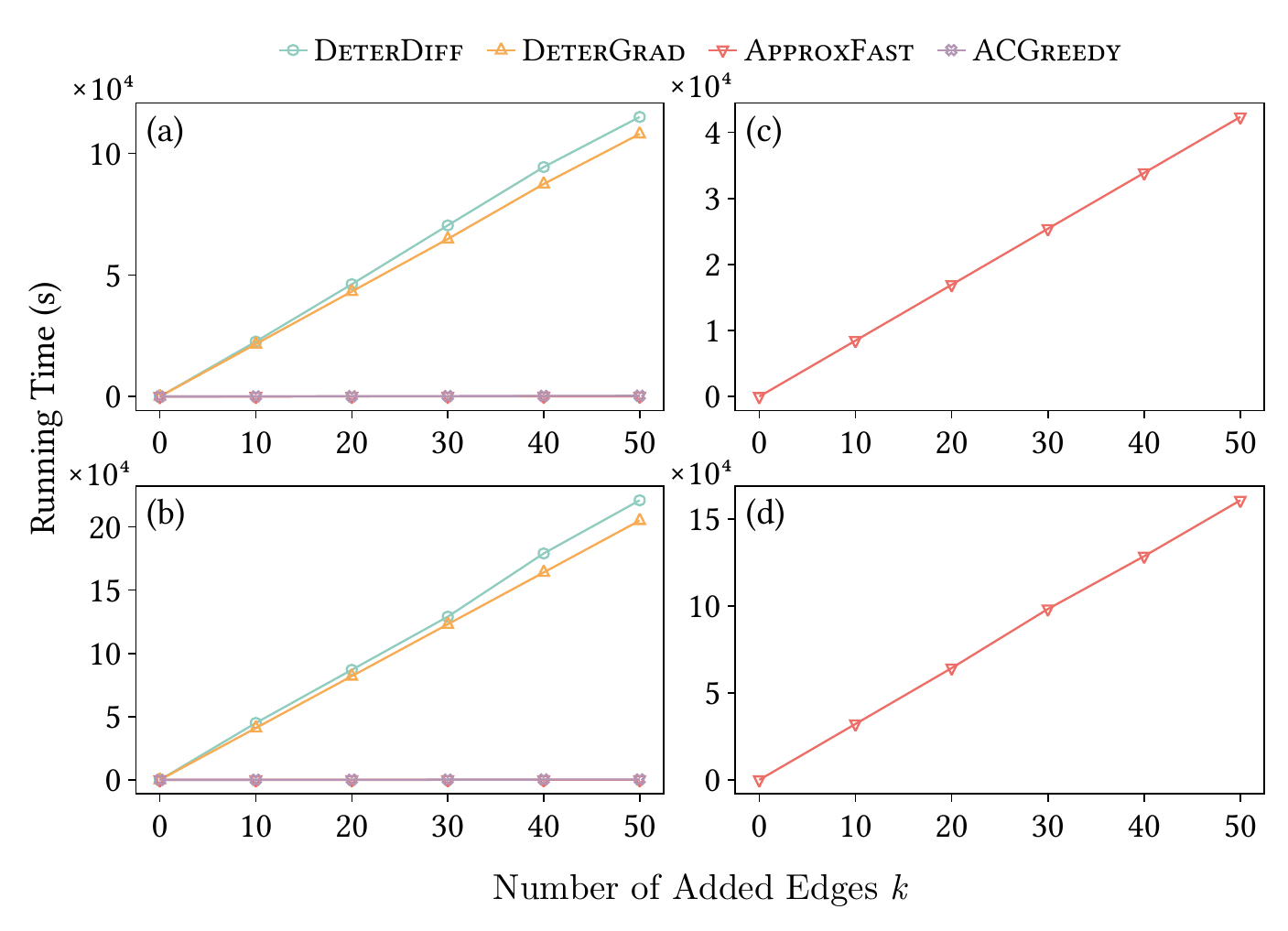}
		\caption{Running time of different algorithms for $k=10,20,\ldots,50$ on four networks: \textsf{Facebook} (a), \textsf{Hepth} (b), \textsf{YoutubeSnap} (c), and \textsf{Livejournal} (d).  \label{f4}}
\end{figure}

\subsection{Efficiency Comparison}

Although our three greedy algorithms perform remarkably well in terms of effectiveness, we will demonstrate that \textsc{ApproxFast} is significantly faster than \textsc{DeterDiff}, \textsc{DeterGrad} and \textsc{ACGreedy}. To illustrate this, we compare the efficiency of these algorithms. Figure~\ref{f4} presents the running time of these algorithms across various networks. For each network, we select $k=50$ edges. Notably, for both \textsf{Facebook} and \textsf{Hepth} networks, consisting of 4,039 and 8,638 nodes respectively,  \textsc{DeterGrad} and \textsc{DeterDiff} take over 30 hours to select $k=50$ edges, making them impractical for networks larger than 10,000 nodes. {
Due to the extreme slowness of \textsc{DeterGrad} and \textsc{DeterDiff}, the curves of \textsc{ACGreedy} and \textsc{ApproxFast} appear to overlap for \textsf{Facebook} and \textsf{Hepth}. However, \textsc{ACGreedy} is in fact much slower than \textsc{ApproxFast}. For $k = 50$, \textsc{ACGreedy} takes 295 seconds on \textsf{Facebook} and 332 seconds on \textsf{Hepth}, while \textsc{ApproxFast} completes in only 118 and 110 seconds, respectively. On the larger \textsf{Douban} network with 154,908 nodes, \textsc{ACGreedy} takes over 10 hours for $k = 50$, whereas \textsc{ApproxFast} finishes within 1 hour. This makes \textsc{ACGreedy} unsuitable for networks with over 150,000 nodes. As a result, \textsc{DeterDiff}, \textsc{DeterGrad}, and \textsc{ACGreedy} are omitted from Figure~\ref{f4} for the \textsf{YoutubeSnap} and \textsf{LiveJournal} networks (Figures~\ref{f4}(c) and (d)), which have over 1 million nodes, as they could not complete within a reasonable time frame.} In stark contrast, \textsc{ApproxFast} shows remarkable efficiency, completing in just 2 minutes on both the \textsf{Facebook} and \textsf{Hepth} networks with $k=50$ edges. For \textsf{YoutubeSnap}, which has more than 1 million nodes and 2 million edges, \textsc{ApproxFast} requires only 11 hours, and for the largest network, with over 4 million nodes and 27 million edges, \textsc{ApproxFast} completes in 44 hours, both with $k=50$. These results underscore the scalability and practicality of \textsc{ApproxFast}.

{Finally, we point out that since the optimal edge selection strategy depends on the optimization objective, and thus a specific edge selection strategy is required for a specific problem. Taking the minimization of the leading eigenvalue of a non-backtracking matrix as an example, if the edges are selected based on the reduction of the spectral radius of the adjacency matrix and the non-backtracking matrix, the resulting edge sets are very different~\cite{zhang2021minimizing}. Therefore, edge selection strategies are problem-dependent. \textsc{ACGreedy} is specially designed to maximize algebraic connectivity. From our experiments, \textsc{ACGreedy} performs significantly worse than our customized methods and fails to identify a more effective edge set.}

\section{Related Work}

\textbf{Resistance Distance Computation.} 
Various algorithms have been developed for estimating resistance distances in large graphs~\cite{HaAkYo16, PeLoYo21, LiLiDa23, YaTa23}. Early methods used random projection~\cite{SpSr08}, while subsequent approaches improved efficiency through random walk sampling and random spanning tree formulations~\cite{PeLoYo21}. More recent work introduced Monte Carlo techniques for faster random walk-based estimations~\cite{YaTa23} and landmark-based sampling methods~\cite{LiLiDa23, liao2024efficient}. Despite extensive studies on resistance distance, its variants, such as biharmonic distance, remain less understood.

\textbf{Biharmonic Distance Computation.} 
Originally introduced by Lipman et al.~\cite{lipman2010biharmonic}, biharmonic distance captures global structural information and is more effective than resistance distance in network influence mining and graph embedding~\cite{YiShLiZh18, kreuzer2021rethinking, black2023understanding}. Several efficient algorithms for biharmonic distance computation have been proposed~\cite{Yi2022BiharmonicDP, zhang2020fast,liu2024fast}, but no methods have been developed to minimize total biharmonic distance through edge addition.

\textbf{Optimization via Edge Addition.} 
{Edge addition has been studied for optimizing various objectives~\cite{xu2023minimizing,then2014more,CrDaSeVe16}, such as increasing the number of spanning trees~\cite{LiPaYi19}, reducing random walk hitting times~\cite{adriaens2023minimizing}, improving centrality measures~\cite{ShYiZh18,medya2018group,DaOlSe19,LiGaWuRoWeQi21}, and minimizing total shortest path distances~\cite{meyerson2009minimizing}.} Other objectives, which are more closely related to network connectivity and robustness, include reducing the network diameter~\cite{FrGaGuMa15,AdGi22}, bolstering algebraic connectivity~\cite{GhBo06}, decreasing the Kirchhoff index~\cite{GhBoSa08,PrKoMe22,kooij2023minimizing}, and increasing robustness~\cite{chan2016optimizing}.
We are the first to propose our problem. None of the existing approaches can be directly translated to the problem of minimizing total biharmonic distance through edge addition. Furthermore, existing approaches may face challenges when applied to large-scale networks. Our work addresses this gap by introducing novel algorithms specifically designed for this challenging optimization problem.

\section{Conclusion}
In this paper, we studied the problem of minimizing the total biharmonic distance in a network through link recommendation. We first demonstrated the monotonicity of the total biharmonic distance, showing that adding new edges to the graph decreases the total biharmonic distance. Based on this property, we formulated a combinatorial optimization problem with broad applicability, including network robustness maximization. The problem is combinatorially complex. We established that while the objective function of the problem is monotone, it is not supermodular. To tackle the problem, we employed greedy heuristics. The first algorithm, \textsc{DeterDiff}, has a provable bound of approximation factor and cubic time complexity. The second algorithm, \textsc{DeterGrad}, uses the gradient rather than the marginal decrease, showing selecting the edge with the largest gradient is an effective strategy. We then leveraged the projection technique, the Laplacian solver, and the convex hull approximation to develop an algorithm with nearly linear time complexity and an error guarantee. Finally, we performed experiments on ten real networks, showing that \textsc{DeterDiff}, \textsc{DeterGrad}, and \textsc{ApproxFast} exhibit comparable effectiveness, with \textsc{ApproxFast} being the most efficient. \textsc{ApproxFast} demonstrates scalability in large graphs containing over 4 million nodes. Future work includes extending our algorithm for other network optimization problems.

\section{Acknowledgments}

The work was supported by the National Natural Science Foundation of China (Nos. 62372112 and 61872093).


\providecommand{\noopsort}[1]{}\providecommand{\singleletter}[1]{#1}

\appendix
\section{APPENDIX}~\label{appendix}
In this section, we give proofs of Lemma~\ref{lem:error} and Theorem~\ref{performance}. We also include additional experimental details.

\subsection{Proof of Lemma~\ref{lem:error}}

\begin{proof}In order to prove the lemma, it suffices to show that for any arbitrary pair of nodes $u$ and $v$,
\begin{equation}
\begin{aligned}\label{13}
&\left|\|\boldsymbol{Z}(\boldsymbol{e}_u-\boldsymbol{e}_v)\|^2-\|\ZZtil(\boldsymbol{e}_u-\boldsymbol{e}_v)\|^2\right| \\=& \left|\|\boldsymbol{Z}(\boldsymbol{e}_u-\boldsymbol{e}_v)\|-\|\ZZtil(\boldsymbol{e}_u-\boldsymbol{e}_v)\|\right| \times  \left|\|\boldsymbol{Z}(\boldsymbol{e}_u-\boldsymbol{e}_v)\|+\|\ZZtil(\boldsymbol{e}_u-\boldsymbol{e}_v)\|\right|\\
\leq & \left(\frac{2\beta}3+\frac{\beta^2}9\right)\|\boldsymbol{Z}(\boldsymbol{e}_u-\boldsymbol{e}_v)\|^2,
\end{aligned}
\end{equation}
which is satisfied if
	\begin{equation}\label{EE221}
		\abs{\|\ZZ (\ee_u-\ee_v)\|-\|\ZZtil (\ee_u-\ee_v)\|} \le
		\frac{\beta}{3}\norm{\ZZ (\ee_u-\ee_v)}.
		\end{equation}

The following arguments can account for this. First, if 
$$
\left|\|\boldsymbol{Z}(\boldsymbol{e}_u-\boldsymbol{e}_v)\|^2 - \|\tilde{\boldsymbol{Z}}(\boldsymbol{e}_u-\boldsymbol{e}_v)\|^2\right| \leq \left(\frac{2\beta}{3} + \frac{\beta^2}{9}\right)\|\boldsymbol{Z}(\boldsymbol{e}_u-\boldsymbol{e}_v)\|^2,
$$
then it follows that 
$$
(1 - \beta)\|\boldsymbol{Z}(\boldsymbol{e}_u-\boldsymbol{e}_v)\|^2 \leq \|\tilde{\boldsymbol{Z}}(\boldsymbol{e}_u-\boldsymbol{e}_v)\|^2 \leq (1 + \beta)\|\boldsymbol{Z}(\boldsymbol{e}_u-\boldsymbol{e}_v)\|^2.
$$
Given the hypothesis 
$$
(1 - \beta) \|\boldsymbol{B} \boldsymbol{L}^{2\dagger} (\boldsymbol{e}_u - \boldsymbol{e}_v)\|^2 \leq \|\boldsymbol{Z} (\boldsymbol{e}_u - \boldsymbol{e}_v)\|^2 \leq (1 + \beta) \|\boldsymbol{B} \boldsymbol{L}^{2\dagger} (\boldsymbol{e}_u - \boldsymbol{e}_v)\|^2,
$$
we have
$$
(1 - \beta)^2 \|\boldsymbol{B} \boldsymbol{L}^{2\dagger} (\boldsymbol{e}_u - \boldsymbol{e}_v)\|^2 \leq \|\tilde{\boldsymbol{Z}} (\boldsymbol{e}_u - \boldsymbol{e}_v)\|^2 \leq (1 + \beta)^2 \|\boldsymbol{B} \boldsymbol{L}^{2\dagger} (\boldsymbol{e}_u - \boldsymbol{e}_v)\|^2,
$$
which directly leads to~\eqref{eqimportant1}.

In turn, if~\eqref{EE221} holds true, then
$$
\|\tilde{\boldsymbol{Z}} (\boldsymbol{e}_u - \boldsymbol{e}_v)\| \leq \left(1 + \frac{3}{\beta}\right) \|\boldsymbol{Z} (\boldsymbol{e}_u - \boldsymbol{e}_v)\|.
$$
Thus, we obtain the bound 
$$
\left|\|\boldsymbol{Z} (\boldsymbol{e}_u - \boldsymbol{e}_v)\| + \|\tilde{\boldsymbol{Z}} (\boldsymbol{e}_u - \boldsymbol{e}_v)\|\right| \leq \left(2 + \frac{\beta}{3}\right)\|\boldsymbol{Z} (\boldsymbol{e}_u - \boldsymbol{e}_v)\|,
$$
which directly leads to~\eqref{13}.

We next prove that~\eqref{EE221} holds true. Let $P$ denote a simple path connecting $u$ and $v$. Applying the triangle inequality along $P$, we obtain
\begin{align*}
	&\abs{\norm{\ZZ(\ee_u-\ee_v)} - \norm{\ZZtil(\ee_u-\ee_v)}}
		\leq \norm{(\ZZ - \ZZtil)(\ee_u-\ee_v)} 
		\\
		\leq & 
		\sum\nolimits_{a\sim b \in P}
		\norm{(\ZZ - \ZZtil) \kh{\ee_a - \ee_b}} \\
        \leq& \kh{n\sum\nolimits_{a\sim b \in P}
			\norm{(\ZZ - \ZZtil) \kh{\ee_a - \ee_b}}^2 }^{1/2}\\
		\leq & n^{1/2} \kh{\sum\nolimits_{a\sim b\in E}
			\norm{(\ZZ - \ZZtil) \kh{\ee_a - \ee_b}}^2}^{1/2} \\
		= & n^{1/2} \norm{(\ZZ - \ZZtil) \BB^\top }_F,
\end{align*}
where the third inequality is derived by Cauchy-Schwarz Inequality.

We now transform the above obtained Frobenius norm $n^{1/2} \| (\ZZ - \ZZtil) \BB^\top\| _F$ into the $\LL$-norm as
		\begin{align*}
		&n^{1/2} \| (\ZZ - \ZZtil) \BB^\top\| _F 
		=n^{1/2} \sqrt{\trace{(\ZZ - \ZZtil) \BB^\top \BB (\ZZ - \ZZtil)^\top}} \\
		= & n^{1/2} \sqrt{\trace{(\ZZ - \ZZtil) \LL (\ZZ - \ZZtil)^\top}} \\
		= & n^{1/2} \sqrt{\sum\nolimits_{i=1}^m
 (\zz_i - \tilde{\zz}_i)^\top \LL (\zz_i - \tilde{\zz}_i)} \\
		\leq &
		n^{1/2} \delta'
		\|\QQ \LL^{\dagger} \BB^\top\|_F,
		\end{align*}
where the inequality is obtained according to~\eqref{laplacain eq} and $\delta'$ is given by~\eqref{delta'}. We continue to simplify $n^{1/2} \delta'
		\|\QQ \LL^{\dagger} \BB^\top\|_F$ as
\begin{align*}
		&n^{1/2} \delta'
		\|\QQ \LL^{\dagger} \BB^\top\|_F = n^{1/2}  \delta'
		\kh{\sum\nolimits_{a\sim b \in E} 
			\norm{\QQ \LL^{\dagger} (\ee_a - \ee_b)}^2}^{1/2}
		\\
		\leq & n^{1/2}  \delta'
		\kh{(1 + \beta)
			\sum\nolimits_{a\sim b \in E}  R_{a,b}  }^{1/2}
		\\
		 =&
             n^{1/2}  \delta'
		\kh{
			(1 + \beta)(n-1)}^{1/2}.
\end{align*}

 On the other hand, we provide a lower bound of $\norm{\ZZ (\ee_u-\ee_v)}^2$ as
 \begin{align*}
\norm{\ZZ (\ee_u-\ee_v}^2
	\geq (1-\beta)(\ee_u-\ee_v)^{\top} \LL^{3\dagger} (\ee_u-\ee_v).
\end{align*}
The eigenvalues of the Laplacian matrix $\LL$ of $G$ are at most $n$, so the non-zero eigenvalues of $\LL^\dagger$ are at least $\frac{1}{n}$ and the non-zero eigenvalues of $\LL^{3\dagger}$ areat least $\frac{1}{n^3}$. Moreover, the vector $\ee_u-\ee_v$ is orthogonal to $\ker \LL$, since $\ker \LL$ is spanned by the all-ones vector. By the Courant-Fischer theorem, which states that for a symmetric matrix $\AA$ with minimal non-zero eigenvalue $\lambda_{\min}$, and for any non-zero vector $\xx \perp \ker \AA$, we have $\lambda_{\min} \leq \frac{\xx^T\AA\xx}{\xx^T\xx}$. Thus, we obtain
$$
(\ee_u-\ee_v)^{\top} \LL^{3\dagger} (\ee_u-\ee_v) \geq \frac{2}{n^3}.
$$
Therefore,  
\begin{align*}
\norm{\ZZ (\ee_u-\ee_v}^2 \geq (1-\beta)\frac{2}{n^3}.
\end{align*}Combining the above-obtained results, it follows that
\begin{small}
    \begin{align*}
    &\quad \frac{
        \abs{ \norm{\ZZ (\ee_u-\ee_v)} -  \norm{\ZZtil(\ee_u-\ee_v)}}
    }{
        \norm{\ZZ (\ee_u-\ee_v)}
    }\\
    &\le
    \delta' n^2((1 + \beta)(n-1))^{1/2}\left(\frac{1
        }{2(1-\beta)}\right)^{1/2}\le \frac{\beta}{3}.
    \end{align*}
\end{small}
This completes the proof.
\end{proof}

\subsection{Proof of Theorem~\ref{performance}}
\begin{proof} We begin by analyzing the time complexity of the algorithm. First, a random matrix \(\boldsymbol{Q}\) is generated, which takes \(O(m \log n / \epsilon^{2})\) time (Line 4). Next, we compute \(\QQ \BB\), which requires \(2m \times 24 \log n \epsilon^{-2} = \tilde{O}(m \epsilon^{-2})\) time, as \(\BB\) contains \(2m\) entries (Lines 5-6). The algorithm proceeds through \(k\) iterative rounds (Lines 7-16), with each round selecting a single edge. During each iteration, the Laplacian solver is applied \(2t\) times, resulting in a total cost of \(2 \times 24 \log n \epsilon^{-2} \times m = \tilde{O}(m \epsilon^{-2})\). Consequently, the points \(\{p_1, p_2, \ldots, p_n\}\) are computed in \(\tilde{O}(m \epsilon^{-2})\) time (Lines 8-12). 

Subsequently, the convex hull of the point set \(P\) is computed using \textsc{ApproxConv}, with a time complexity of \(O(n \log n \epsilon^{-2} l)\) (Line 13). The maximum distance between points on the convex hull is then determined in \(O(l^2 \log n \epsilon^{-2})\) time (Line 14). Finally, the solution, the graph, and the Laplacian matrix are updated based on the identified results (Lines 15-16). Thus, the overall time complexity of the \textsc{ApproxFast} algorithm is \(\tilde{O}(k(nl + m)/\epsilon^2)\).\end{proof}

\subsection{Experimental Details}\label{appendix:expdetails}
\subsubsection{Datasets Used in Experiments}
Table~\ref{tab:1} provides the details of each network.

\begin{table}[H]
        \centering
 \setlength{\tabcolsep}{1.2mm} 
 \caption{Datasets used in experiments.}\label{tab:1}
 \begin{tabular}{ccccccc}
  \hline
            Type &Networks &$n$ &$m$   \\ 
            \hline
            \multirow{4}{*}{\parbox[t]{14mm}{\centering Small Network}}
            &\textsf{Tribes} & 16 & 58  \\
            &\textsf{Karate}& 34 & 78   \\
            &\textsf{Dolphins}  & 62 & 159 \\
            &\textsf{Moreno} & 64 &243 \\
            \hline
            \multirow{4}{*}{\parbox[t]{14mm}{\centering Medium Network}}
            &\textsf{Hamster}& 2,000 & 16,097 \\
            &\textsf{Facebook} & 4,039 & 88,234  \\
            &\textsf{Grqc} & 4,158 & 13,422   \\
            &\textsf{Hepth} & 8,638 & 24,806  \\
            \hline
         \multirow{2}{*}{\parbox[t]{14mm}{\centering Large Network}}
            &\textsf{YoutubeSnap} &1,134,890 &2,987,624 \\
            &\textsf{Livejournal} &4,033,137 &27,933,062  \\
            \hline
 \end{tabular}
\end{table}

{

\subsubsection{Parameter Sensitivity Analysis}
We further examine the influence of varying error parameter $\epsilon$ on the accuracy and efficiency of \textsc{ApproxFast}. Specifically, we conduct a series of experiments across several real-world networks, adjusting $\epsilon$ within the range of 0.2 to 0.5. According to the results in Table~\ref{parameter}, the running time of \textsc{ApproxFast} increases approximately in proportion to $\epsilon^{-2}$, which empirically supports its theoretical time complexity of $\tilde{O}(k(nl+m)/\epsilon^2)$. In addition, Table~\ref{parameter} reveals that the results returned by \textsc{ApproxFast} remain similar as $\epsilon$ changes.}

\begin{table}[H]
	\centering
	\caption{The running time (seconds, $s$) and $B(G)/n$ returned by \textsc{ApproxFast} when $k=10$ on real-world networks.}\label{parameter}
	\fontsize{6.8}{7}\selectfont			
\begin{tabular}{m{0.7cm}m{0.5cm}<{\centering}m{0.4cm}<{\centering}m{0.4cm}<{\centering}m{0.5cm}<{\centering}m{0.6cm}<{\centering}m{0.6cm}<{\centering}m{0.6cm}<{\centering}m{0.6cm}<{\centering}}
\toprule
\multirow{2}{*}{Networks} 
& \multicolumn{4}{c}{Running time $(s)$} 
& \multicolumn{4}{c}{$B(G)/n$} \\
\cmidrule(lr){2-5} \cmidrule(lr){6-9}
& $\epsilon{=}0.2$ & $\epsilon{=}0.3$ & $\epsilon{=}0.4$ & $\epsilon{=}0.5$
& $\epsilon{=}0.2$ & $\epsilon{=}0.3$ & $\epsilon{=}0.4$ & $\epsilon{=}0.5$ \\
\midrule
\textsf{Hamster} & 40 & 16 & 10 &6 &588 &592 &594 &604 \\
\textsf{Facebook} & 91 & 42 & 24 &16 &2178 &2198 &2215 &2245 \\
\textsf{Grqc} & 71 &32 &19 &12 &9900 &9988 &10011 &10049 \\
\textsf{Hepth} & 80 &35 &23 &14 &15452 &15497 &15532 &15582 \\
\makecell[l]{\textsf{Youtube}\\\textsf{Snap}} & 32863 &14036 &8465 &5418 &4470080 &4471300 &4472020 &4473103\\[-1mm]
\bottomrule
\end{tabular}
\end{table}

\end{document}